\documentclass[10pt]{article}

\usepackage{geometry}	
\usepackage{amsmath}
\usepackage{amssymb}
\usepackage{amsthm}
\usepackage{mathrsfs}
\usepackage{enumerate}
\usepackage{epsf}
\usepackage{psfrag}
\usepackage{dsfont}
\usepackage{hyperref}
\usepackage{cleveref}
\usepackage{graphicx}
\usepackage{empheq}
\DeclareGraphicsExtensions{.eps,.art,.ART,.ps}
\allowdisplaybreaks
\newcommand{\cL}{\mathcal{L}}

\newcommand{\bbR}{\mathbb{R}}      
\newcommand{\bbZ}{\mathbb{Z}}      

\newcommand*\widefbox[1]{\fbox{\hspace{0em}#1\hspace{0em}}}
\newtheorem{Thm}{Theorem}[section]

\begin{document}
\title{Rotating elastic string loops in flat and black hole spacetimes: stability, cosmic censorship and the Penrose process}
\author{Jos\'e Nat\'ario$^{1}$, Leonel Queimada$^{1,2}$ and Rodrigo Vicente$^{1,3}$\\
{\small $^1$ CAMGSD, Departamento de Matem\'{a}tica, Instituto Superior T\'{e}cnico, Universidade de Lisboa, Portugal}\\ 
{\small $^2$ Perimeter Institute for Theoretical Physics, Waterloo, Ontario N2L 2Y5, Canada}\\
{\small $^3$ CENTRA, Departamento de F\'{\i}sica, Instituto Superior T\'{e}cnico, Universidade de Lisboa, Portugal}}
\date{}
\maketitle
%
%
\begin{abstract}
We rederive the equations of motion for relativistic strings, that is, one-dimensional elastic bodies whose internal energy depends only on their stretching, and use them to study circular string loops rotating in the equatorial plane of flat and black hole spacetimes. We start by obtaining the conditions for equilibrium, and find that: (i) if the string's longitudinal speed of sound does not exceed the speed of light then its radius when rotating in Minkowski's spacetime is always larger than its radius when at rest; (ii) in Minkowski's spacetime, equilibria are linearly stable for rotation speeds below a certain threshold, higher than the string's longitudinal speed of sound, and linearly unstable for some rotation speeds above it; (iii) equilibria are always linearly unstable in Schwarzschild's spacetime. Moreover, we study interactions of a rotating string loop with a Kerr black hole, namely in the context of the weak cosmic censorship conjecture and the Penrose process. We find that: (i) elastic string loops that satisfy the null energy condition cannot overspin extremal black holes; (ii) elastic string loops that satisfy the dominant energy condition cannot increase the maximum efficiency of the usual particle Penrose process; (iii) if the dominant energy condition (but not the weak energy condition) is violated then the efficiency can be increased. This last result hints at the interesting possibility that the dominant energy condition may underlie the well known upper bounds for the efficiencies of energy extraction processes (including, for example, superradiance).
\end{abstract}
\tableofcontents
%
%
%
%
%
\section{Introduction}
The modern formulation of relativistic elasticity is due to Carter and Quintana~\cite{CQ72}, and has been further developed in \cite{Maugin78, KM92, Tahvildar98, BS02, KS03} (see also \cite{Wernig06} and references therein). Its applications range from the computation of the speeds of sound for various relativistic solids \cite{Carter73, Bento85} to the study of elastic equilibrium states \cite{Park00, KS04, ABS08, ABS09, BCV10, AC14, BCMV14} as well as dynamical settings \cite{Magli97, Magli98, CH07}.

In this paper we discuss relativistic strings, that is, one-dimensional elastic bodies whose internal energy depends only on their stretching, first studied by Carter \cite{Carter89, Carter89b} as models for superconducting cosmic strings.\footnote{There exists a wealth of literature dealing with (possibly current carrying) cosmic and Nambu-Goto strings moving in flat and curved spacetimes, see for instance \cite{Anderson03} and references therein. Most of these works correspond to specific choices of elastic laws, as opposed to the full generality which we adopt here.} More precisely, we consider axially symmetric string loops rotating in the equatorial plane of the Minkowski, Schwarzschild and Kerr spacetimes. We determine the equilibrium configurations, investigate their linear stability, and use the interaction between the string loops and the black hole to study cosmic censorship and the Penrose process. This analysis is particularly interesting because elastic strings, in a way, bridge the gap between particles and fields: they are extended objects without being completely spread throughout spacetime.

The organization of the paper is as follows: in Section \ref{section1} we re-derive the equations of motion for elastic strings starting from a Lagrangian density, and rewrite these equations as conservation of energy-momentum along the worldsheet plus the so-called {\em generalized sail equations} (the vanishing of the contraction between the energy-momentum tensor and the extrinsic curvatures of the worldsheet). We also show how to obtain conserved quantities for the string's motion in spacetimes with symmetries, compute the string's longitudinal and transverse speeds of sound, and give examples of simple elastic laws and their properties. In Section \ref{section2} we determine the equilibrium configurations for axially symmetric string loops rotating in the equatorial plane of the Kerr solution, and analyze in detail the particular cases of the Minkowski and Schwarzschild spacetimes. We prove that in flat spacetime the radius of a rotating string loop is always larger than its radius when at rest provided that the string's longitudinal speed of sound does not exceed the speed of light (this conclusion is not trivial due to the length contraction effect, and it is in fact false when string's longitudinal speed of sound does exceed the speed of light). Moreover, we show that in flat spacetime equilibria are linearly stable for rotation speeds below a certain threshold, higher than the string's longitudinal speed of sound, and linearly unstable for some rotation speeds above it, confirming and extending the results in \cite{CM93, Martin94}. Finally, we find that equilibria are always linearly unstable in Schwarzschild's spacetime. In Section \ref{section3} we consider the motion of an axially symmetric string loop rotating in the equatorial plane of an extremal Kerr black hole, and show that the loop cannot overspin the black hole provided that it satisfies the null energy condition. We also consider a Penrose process for elastic string loops, where an incoherent loop (dust ring) falls from infinity and breaks up inside the ergoregion into an elastic string loop plus a second dust ring which escapes to infinity. We find that if the elastic loop satisfies the dominant energy condition then the efficiency of this process cannot be higher than the maximum efficiency for the usual particle Penrose process; however, if the dominant energy condition (but not the weak energy condition) is violated then the efficiency can be increased. This last result hints at the interesting possibility that the dominant energy condition may underlie the well known upper bounds for the efficiencies of energy extraction processes (including, for example, superradiance).

We follow the conventions of \cite{MTW73, W84}; in particular, we use a system of units for which $c=G=1$. Greek letters $\alpha, \beta, \ldots$ represent spacetime indices, running from $0$ to $n$, small case Latin letters $i, j, \ldots$ represent spatial indices, running from $1$ to $n$ (or sometimes from $2$ to $n$), and capital Latin letters $A, B, \ldots$ represent indices in the string's worldsheet, taking the values $0$ and $1$. We used {\sc Mathematica} for symbolic and numerical computations, and also to produce various plots.
%
%
%
\section{Elastic string theory}\label{section1}
In this section, we use a variational approach to re-derive the equations of motion of an elastic string, for the reader's convenience and also to fix notation (see \cite{SPG87, CHT87, Larsen93, Carter92, Carter01, Anderson03, Carter11} for similar or related derivations). These equations are then shown to be equivalent to the conservation of an energy-momentum tensor defined on the worldsheet plus the vanishing of the contraction between the energy-momentum tensor and the extrinsic curvatures of the worldsheet (dubbed {\em generalized sail equations} in \cite{CM93}). We obtain the conserved quantities associated to Killing vector fields (as required by Noether's theorem), and compute the string's longitudinal and transverse speeds of sound, originally given in \cite{Carter89b}. Finally, we present examples of simple elastic laws and their properties.
\subsection{Lagrangian density}

We model a string moving on a $(n+1)$-dimensional spacetime $(M,g)$ by an embedding $X:\bbR \times I \to M$, where $I \subset \bbR$ is an interval labeling the points of the string.\footnote{In the case of string loops we identify the endpoints of $I$ to obtain an embedding $X:\bbR \times S^1 \to M$.} Thus, the curve $\tau \mapsto X(\tau,\lambda)$ is the worldline of the point of the string labeled by $\lambda \in I$. We assume that the parameter $\lambda \in I$ is the arclength in the string's unstretched configuration. 

The embedding $X$ induces a metric
\begin{equation}\label{metric}
\boxed{h_{AB} = g_{\mu\nu}(X) \partial_A X^\mu \partial_B X^\nu}
\end{equation}
on $\bbR \times I$, which we identify with its image $\Sigma = X(\bbR \times I)$ (sometimes called the string's {\em worldsheet}). If we choose a local orthonormal frame $\{E_0, E_1\}$ tangent to $\Sigma$ such that $E_0$ is the $4$-velocity of the string's particles, we must have
\begin{equation}
\begin{cases}
\displaystyle \frac{\partial X}{\partial \tau} = \alpha E_0 \\ \\
\displaystyle \frac{\partial X}{\partial \lambda} = \beta E_0 + \sigma E_1
\end{cases},
\end{equation}
for some smooth local functions $\alpha, \beta, \sigma$. Note that $|\sigma|$ is the factor by which the string is stretched according to an observer comoving with it, since $E_1$ gives the direction of simultaneity for such an observer. The components of the induced metric are then
\begin{equation}
(h_{AB}) = 
\left(
\begin{matrix}
-\alpha^2 & -\alpha \beta \\
- \alpha \beta & - \beta^2 + \sigma^2
\end{matrix}
\right),
\end{equation}
and so
\begin{equation} \label{det}
h \equiv \det(h_{AB}) = - \alpha^2 \sigma^2 = h_{00} \sigma^2.
\end{equation}
Defining the {\em number density} $n = \frac1{|\sigma|}$, we then have
\begin{equation} \label{n^2}
\boxed{n^2 = \frac{h_{00}}{h}}\,\,.
\end{equation}
To obtain the string's equations of motion we must choose an action
\begin{equation}
S = \int_{\bbR \times I} \cL(X,\partial X) \, d\tau d\lambda.
\end{equation}
For an elastic string whose internal energy density $\rho$ depends only on its stretching, $\rho = F(n^2)$, the Lagrangian density is (see for instance \cite{Bento85, KM92})
\begin{equation} \label{Lagrangian}
\cL = F(n^2) \sqrt{-h},
\end{equation}
where $h \equiv \det(h_{AB})$ and $n^2$ are given as functions of $(X, \partial X)$ from equations \eqref{metric} and \eqref{n^2}. Although it is not obvious, this Lagrangian density reduces to the usual Newtonian Lagrangian density for an elastic string in the appropriate limit, as shown in Appendix~\ref{appendixA}.

\subsection{Equations of motion} 

To find the equations of motion we must compute the variation $\delta \cL$ of the Lagrangian density resulting from a variation $\delta X$ of the embedding. Using the well-known formula for the variation of the determinant of the metric,
\begin{equation}
\delta h = h h^{AB} \delta h_{AB}
\end{equation}
(see for instance \cite{W84}), we obtain
\begin{equation}
\delta \cL = \left( F'(n^2) \frac{h_{00} h^{AB} - \delta_0^A \delta_0^B}{ \sqrt{-h}} + \frac12 F(n^2) \sqrt{-h} \, h^{AB} \right) \delta h_{AB}.
\end{equation}
By analogy with the energy-momentum tensor in general relativity, we define the string's energy-momentum tensor $T^{AB}$  by the relation
\begin{equation}
\delta \cL = - \frac12 \sqrt{-h} \, T^{AB} \delta h_{AB}
\end{equation}
(see for instance \cite{W84}). We will provide a better justification for this choice after obtaining the equations of motion. Note that
\begin{equation} \label{TAB}
\boxed{T^{AB} = 2 n^2 F'(n^2) U^A U^B + \left[2 n^2 F'(n^2) - F(n^2)\right] h^{AB}}\,\,,
\end{equation}
where $U^A$ is the four-velocity of the string's particles. Therefore, the string's energy density $\rho$ and the string's pressure $p$ are given by\footnote{Since the worldsheet is two-dimensional, this pressure is actually a force (tension or compression of the string).}
\begin{equation} \label{rhop}
\boxed{\rho = F(n^2)} \,\, , \qquad \boxed{p = 2 n^2 F'(n^2) - F(n^2)} \,\, .
\end{equation}
To derive the equations of motion, we note that
\begin{equation}
\delta h_{AB} = \partial_\alpha g_{\mu\nu} \delta X^\alpha \partial_A X^\mu \partial_B X^\nu + 2 g_{\mu\nu} \partial_A X^\mu \partial_B \delta X^\nu,
\end{equation}
and so
\begin{equation}
- 2 \delta \cL = \left[ \sqrt{-h} \, T^{AB} \partial_\alpha g_{\mu\nu} \partial_A X^\mu \partial_B X^\nu - \partial_B \left( 2 \sqrt{-h} \, T^{AB} g_{\mu\alpha} \partial_A X^\mu \right)\right] \delta X^\alpha + \partial_B \left( 2 \sqrt{-h} \, T^{AB} g_{\mu\alpha} \partial_A X^\mu \delta X^\alpha \right).
\end{equation}
Discarding the total divergence in Hamilton's principle
\begin{equation}
\delta \int_{\bbR \times I} \cL(X,\partial X) \, d\tau d\lambda = 0,
\end{equation}
we obtain the equations of motion in the form
\begin{equation} \label{motion}
\boxed{\frac1{\sqrt{-h}}\partial_B \left( \sqrt{-h} \, T^{AB}\partial_A X^\alpha \right) + T^{AB} \Gamma^\alpha_{\mu\nu} \partial_A X^\mu \partial_B X^\nu = 0}\,\,.
\end{equation}
These equations are closely related to the harmonic map/wave map/nonlinear sigma model equations (see for instance \cite{Jost02, Tao06, Rendall08}). 
\subsection{Adapted coordinates} 

To better understand the equations of motion, we extend the local coordinates $(x^A)=(\tau,\lambda)$ on the worldsheet $\Sigma$ to a local coordinate system $(x^A,x^i)$ defined on a neighborhood of $\Sigma$ in the following way: we choose an orthonormal frame $\{E_2, \ldots, E_n\}$ for the normal bundle of the worldsheet, and parameterize by $(x^A, x^i)$ the point $\exp_p(x^i E_i)$, where $\exp_p$ is the geodesic exponential map and $p \in \Sigma$ is the point with coordinates $(x^A)$.\footnote{In the case of string loops we take the coordinate $x^1$ (which extends $\lambda$) to be periodic.} The worldsheet is given in these coordinates by $x^i=0$, and the spacetime metric by
\begin{equation}
g = g_{AB} dx^A dx^B + 2 g_{Ai} dx^A dx^i + g_{ij} dx^i dx^j.
\end{equation}
Note that on $\Sigma$ we have
\begin{equation}
g_{|_{\Sigma}} = h_{AB} dx^A dx^B + \delta_{ij} dx^i dx^j.
\end{equation}
The tensor
\begin{equation}
K^i_{AB} = \frac12\partial_i g_{AB},
\end{equation}
defined on $\Sigma$, is called the {\em extrinsic curvature} (or {\em second fundamental form}) of $\Sigma$ in the direction of $E_i$. It easily seen that on $\Sigma$
\begin{equation}
\Gamma^i_{AB} = -K^i_{AB}
\end{equation}
and
\begin{equation}
\Gamma^C_{AB} = \overline{\Gamma}^C_{AB},
\end{equation}
where $\overline{\Gamma}^C_{AB}$ are the Christoffel symbols for the Levi-Civita connection $\overline{\nabla}$ of $h_{AB}$. In this coordinate system, the embedding is simply given by $(X^A,X^i)=(x^A,0)$, and so we can write the first two components of equation~\eqref{motion} as
\begin{equation} \label{conservation}
\frac1{\sqrt{-h}}\partial_B \left( \sqrt{-h}T^{BC} \right) + T^{AB} \overline{\Gamma}^C_{AB} = 0,
\end{equation}
and the last $n-1$ as the so-called {\em generalized sail equations}
\begin{equation} \label{constraints}
\boxed{T^{AB} K^i_{AB} = 0}\,\,.
\end{equation}
Using the well-known formula (see for instance \cite{W84})
\begin{equation}
\partial_B \log \sqrt{-h} = \overline{\Gamma}^A_{BA},
\end{equation}
equation~\eqref{conservation} is easily seen to be equivalent to
\begin{equation} \label{conservation2}
\boxed{\overline{\nabla}_B T^{BC} = 0}\,\,.
\end{equation}
This justifies the choice of $T^{AB}$ as the string's energy-momentum tensor. 

The equations of motion of the string can then be understood as constraints of the geometry of the worldsheet, given by \eqref{constraints}, plus conservation of energy-momentum, given by \eqref{conservation2}.\footnote{These equations occur for other extended objects such as branes \cite{Carter92, Carter01, Carter11} and blackfolds \cite{EHNO10}.} For Nambu-Goto strings, for instance, where $T_{AB}$ is proportional to $h_{AB}$, the constraints are the condition that the worldsheet is a minimal surface, and the conservation equation is automatically satisfied.

Interestingly, the component of the conservation equations along $U^A$ is always trivial, even in the general case. Indeed, we have the identity
\begin{equation}
\overline{\nabla}_A (n \, U^A) = \frac1{\sqrt{-h}} \partial_A (\sqrt{-h} \, n \, U^A) = \frac1{\sqrt{-h}} \partial_A \delta_0^A = 0
\end{equation}
({\em particle number conservation}), and so
\begin{equation}
-U_B \overline{\nabla}_A T^{AB} = \overline{\nabla}_U \rho + (\rho + p)\overline{\nabla}_A U^A = 2n F'(n^2) \overline{\nabla}_U n + 2n^2 F'(n^2) \overline{\nabla}_A U^A = 0.
\end{equation}

\subsection{Conserved quantities} 

If $(M,g)$ admits a Killing vector field $\xi$,
\begin{equation}
\nabla_{(\mu} \xi_{\nu)} = 0,
\end{equation}
then in the coordinates above we have
\begin{equation}
\partial_{(A} \xi_{B)} + \overline{\Gamma}^C_{AB} \xi_{C} + \Gamma^i_{AB} \xi_i = 0 \Leftrightarrow \overline\nabla_{(A} \xi_{B)} 
= K^i_{AB} \xi_i,
\end{equation}
that is, the projection of $\xi$ on $T\Sigma$ is not, in general, a Killing vector field of $h_{AB}$. Nevertheless,
\begin{equation}
\overline\nabla_{A} (T^{AB} \xi_{B}) =  T^{AB} K^i_{AB} \xi_i = 0,
\end{equation}
in view of \eqref{constraints}, that is, the vector field
\begin{equation}
\boxed{j^A = T^{AB} \xi_{B}}
\end{equation}
is divergenceless on $\Sigma$. As a consequence, the quantity
\begin{equation}
E^\xi = \int_{\{\tau=\text{constant}\}} j^A \nu_A \sqrt{h_{11}} \, d\lambda 
\end{equation}
is conserved, where 
\begin{equation}
\nu_A = \frac{\delta^0_A}{\sqrt{-h^{00}}} = \frac{\sqrt{-h}}{\sqrt{h_{11}}} \, \delta^0_A 
\end{equation}
is the past-pointing normal to the spacelike curve $\{\tau=\text{constant}\}$. In other words, we have the conserved quantity
\begin{equation}\label{conserved_quantity}
\boxed{E^\xi = \int_{\{\tau=\text{constant}\}} j^0 \sqrt{-h} \, d\lambda}\,\,.
\end{equation}
\subsection{Speeds of sound}
The speeds of local perturbations traveling on a string can be obtained by linearizing the equations of motion about a (possibly stretched) stationary string in Minkowski spacetime, aligned, say, with the $x$-axis. This corresponds to taking terms up to quadratic order in the Lagrangian obtained from the embedding
\begin{equation}
\begin{cases}
t(\tau, \lambda) = \tau \\
x(\tau, \lambda) = {n_0}^{-1} \lambda + \delta x(\tau, \lambda) \\
y(\tau, \lambda) = \delta y(\tau, \lambda) \\
z(\tau, \lambda) = \delta z(\tau, \lambda)
\end{cases} .
\end{equation}
To compute $h_{00}$ and $h$ to quadratic order it suffices to use the approximation
\begin{equation}
(h_{AB}) = 
\left(
\begin{matrix}
- 1 + \delta \dot{x}^2 + \delta \dot{y}^2 + \delta \dot{z}^2 & {n_0}^{-1} \delta \dot{x} \\
{n_0}^{-1} \delta \dot{x} & {n_0}^{-2} + 2 {n_0}^{-1} \delta x' + {\delta x'}^2 + {\delta y'}^2 + {\delta z'}^2  
\end{matrix}
\right),
\end{equation}
where $\,\,\dot{} \equiv \frac{\partial}{\partial \tau}$ and $\,\,'\equiv \frac{\partial}{\partial \lambda}$, so that
\begin{equation}
n^2 = \frac{h_{00}}{h} = {n_0}^2 \left( 1 - \delta \dot{x}^2 - 2 {n_0} \delta x' + 3 {n_0}^2 {\delta x'}^2 - {n_0}^2 {\delta y'}^2 - {n_0}^2 {\delta z'}^2\right).
\end{equation}
Using the Taylor formula to second order,
\begin{equation}
F(n^2) = F({n_0}^2) + F'({n_0}^2) (n^2 - {n_0}^2) + \frac12 F''({n_0}^2) (n^2 - {n_0}^2)^2,
\end{equation}
and discarding constants and total divergences, we finally obtain
\begin{align}
\cL = F(n^2) \sqrt{-h} = \, & F'({n_0}^2) {n_0}^2 \left[ \left( 2{n_0}^4 \frac{F''({n_0}^2)}{F'({n_0}^2)} + {n_0}^2 \right) {\delta x'}^2 - \delta \dot{x}^2 \right] \nonumber \\
& + \frac12 F({n_0}^2) {n_0}^2 \left[ \left( {n_0}^2 - 2{n_0}^4 \frac{F'({n_0}^2)}{F({n_0}^2)} \right) {\delta y'}^2 - \delta \dot{y}^2 \right] \\
& + \frac12 F({n_0}^2) {n_0}^2 \left[ \left( {n_0}^2 - 2{n_0}^4 \frac{F'({n_0}^2)}{F({n_0}^2)} \right) {\delta z'}^2 - \delta \dot{z}^2 \right]. \nonumber
\end{align}
Therefore $\delta x$ satisfies the wave equation in the coordinates $(\tau,\lambda)$ with wave speed
\begin{equation}
c' = n_0 \sqrt{2{n_0}^2 \frac{F''({n_0}^2)}{F'({n_0}^2)} + 1},
\end{equation}
whereas $\delta y$ and $\delta z$ satisfy the wave equation with wave speed
\begin{equation}
s' = n_0 \sqrt{1 - 2{n_0}^2 \frac{F'({n_0}^2)}{F({n_0}^2)}}.
\end{equation}
Since $\lambda = n_0 x$ for the stretched string, we see that the physical speed of sound for longitudinal waves is
\begin{equation} \label{c^2}
\boxed{c = \sqrt{2n^2 \frac{F''(n^2)}{F'(n^2)} + 1} = \sqrt{\frac{dp}{d\rho}}}\,\,,
\end{equation}
the same expression as the speed of sound for a perfect fluid (see for instance \cite{Christodoulou97}), whereas the speed of sound for transverse waves is given by
\begin{equation}
\boxed{s = \sqrt{1 - 2n^2 \frac{F'(n^2)}{F(n^2)}} = \sqrt{-\frac{p}{\rho}}}\,\,,
\end{equation}
generalizing the well known classical result. 

As noted in \cite{Carter89b}, a necessary condition for the stability of the stretched string is that $c$ and $s$ be real,\footnote{Otherwise there would exist exponentially growing modes in the limit of small wavelengths.} that is, $\frac{dp}{d\rho}\geq 0$ and $p \leq 0$.
\subsection{The choice of elastic law}\label{warm}
There are many possible choices for the so-called {\em elastic law} $\rho = F(n^2)$, each corresponding to a different kind of elastic string. Some important examples are (for a given constant energy density $\rho_0>0$ of the unstretched string):
\begin{itemize}
\item
\underline{(Non-prestressed) strings with constant longitudinal speed of sound $c\geq 0$}:  These are given by
\begin{equation} \label{rigid}
\rho = \frac{\rho_0}{c^2+1} \left(n^{c^2+1} + c^2\right), \qquad \text{yielding} \qquad p = \frac{\rho_0 c^2}{c^2+1} \left(n^{c^2+1} - 1\right).
\end{equation}
For $c=1$ we obtain the so-called rigid string \cite{N14}, and for $c=0$ we have an incoherent (dust) string.
\item
\underline{Strings with constant transverse speed of sound $s \geq 0$}: This model corresponds to choosing
\begin{equation}
\rho = \rho_0 n^{1-s^2}, \qquad \text{ whence } \qquad p = - s^2 \rho.
\end{equation}
For $s=1$ we obtain the Nambu-Goto string, and for $s=0$ we again have a dust string.
\item
\underline{``Warm'' cosmic string model with mass parameter $m \geq 0$} \cite{Carter90a, Vilenkin90}: These are obtained from
\begin{equation}
\rho = \sqrt{({\rho_0}^2 - m^4)n^2 + m^4}, \qquad \text{ implying } \qquad p = - \frac{m^4}{\rho}
\end{equation}
(with $m^2 < \rho_0$). In this case, the longitudinal and transverse speeds of sound coincide. For $m=0$ we again have a dust string.
\end{itemize}
Depending on the elastic law, the string may have different properties:
\begin{itemize}
\item
\underline{Existence of a relaxed configuration}: If the pressure is zero when the string is not stretched nor compressed (that is, if the string is not prestressed) then $F$ must satisfy
\begin{equation} \label{p(1)=0}
2F'(1) = F(1).
\end{equation}
Of the three models above, only the first satisfies this condition.
\item
\underline{Weak energy condition}: The weak energy condition $\rho \geq 0$ and $\rho + p \geq 0$ is equivalent to
\begin{equation} \label{weak}
F(n^2) \geq 0 \text{ and } F'(n^2) \geq 0.
\end{equation}
In particular, if the string satisfies the weak energy condition then $\rho$ is a nondecreasing function of $n^2$. All the models above satisfy this condition.
\item
\underline{Dominant energy condition}: The dominant energy condition $\rho \geq p \geq -\rho$ is equivalent to
\begin{equation} \label{dominant}
F(n^2) \geq n^2 F'(n^2) \geq 0.
\end{equation}
As is well known, if the string satisfies the dominant energy condition then it also satisfies the weak energy condition. Of the three models above, only the first two satisfy the dominant energy condition, and only for $c \leq 1$ and $s \leq 1$.\footnote{It is clear that if an elastic string satisfies the dominant energy condition then its transverse speed of sound cannot exceed the speed of light.}
\item
\underline{Well defined longitudinal speed of sound}: If the longitudinal speed of sound is well defined then from \eqref{c^2} we must have $F'(n^2)\neq 0$ and $\frac{dp}{d\rho} \geq 0$. Of the three models above, only the first and the third satisfy this condition.\footnote{Technically, the second model also satisfies this condition in the trivial case $s=0$.} If the string also satisfies the weak energy condition then $\rho$ is a strictly increasing function of $n^2$, and hence $p$ is a nondecreasing function of $n^2$.
\end{itemize}
%
%
\section{Rigidly rotating string loops}\label{section2}
In this section we find the equilibrium conditions for axially symmetric string loops rotating in the equatorial plane of a Kerr solution with mass $M\geq 0$ and angular momentum $Ma$. In the Minkowski case $M=a=0$, we prove that the radius of a rotating string loop is always larger than its radius when at rest provided that the string's longitudinal speed of sound does not exceed the speed of light. This is a nontrivial statement because the length contraction effect must be taken into account; it is in fact false when string's longitudinal speed of sound does exceed the speed of light. We also show that equilibria are linearly stable for rotation speeds below a certain threshold, higher than the string's longitudinal speed of sound, and linearly unstable for some rotation speeds above it, confirming and extending the results of Carter and Martin \cite{CM93, Martin94}. Finally, we find that equilibria are always linearly unstable in the Schwarzschild case $M>0$, $a=0$.

Equilibrium configurations of Nambu-Goto and cosmic strings in black hole spacetimes have been previously analyzed in \cite{FSZ89, CF89, Carter90, CFH91, Larsen91, FS04, KF08, AA09}, and rigidly rotating strings were considered in \cite{FHV97, OIKNS08, II12, KIT16}. Stability issues for Nambu-Goto and cosmic strings have been discussed in \cite{Carter89b, FS91, LS93, Larsen94a, Larsen94b, BS09}.
\subsection{Equilibrium conditions}
Let us start by considering the Kerr solution with mass $M\geq 0$ and angular momentum $Ma$, given in Boyer-Lindquist coordinates by
\begin{equation}\label{K_metric}
ds^2 = - \left( 1 - \frac{2Mr}{\rho^2} \right) dt^2 + \frac{\rho^2}{\Delta} dr^2 + \rho^2 d \theta^2 + \left(r^2 + a^2 + \frac{2Mra^2}{\rho^2} \sin^2 \theta\right) \sin^2 \theta d\varphi^2 - \frac{4Mra\sin^2 \theta}{\rho^2} dt d\varphi,
\end{equation}
where $\rho^2 = r^2 + a^2 \cos^2 \theta$ and $\Delta = r^2-2Mr+a^2$. As explained above, a string in this geometry is described by an embedding, which we take to be of the form
\begin{equation} \label{embed}
\begin{cases}
t(\tau,\lambda)=\tau \\
r(\tau,\lambda) = R \\
\theta(\tau,\lambda) = \frac{\pi}2 \\
\varphi(\tau,\lambda) = \Omega \tau +  k \lambda
\end{cases},
\end{equation}
where $R$, $\Omega$ and $k$ are constants and $\lambda \in \left[0,\frac{2\pi}{k}\right]$. This embedding describes an axially symmetric elastic string loop rotating in the equatorial plane with angular velocity $\Omega$, in equilibrium at a constant radius $R$. From the definition of $\lambda$, it follows that the radius of the unstretched string loop in flat spacetime is
\begin{equation}
\boxed{R_0 = \frac1k} \,\, .
\end{equation}
The metric induced by the embedding on the worldsheet is
\begin{equation}
(h_{AB}) = 
\left(
\begin{matrix}
-V +\Sigma \Omega^2-\frac{4aM\Omega}{R} & -\frac{2aMk}{R}+\Sigma \Omega k\\
-\frac{2aMk}{R}+\Sigma \Omega k & \Sigma k^2 
\end{matrix}
\right),
\end{equation}
where
\begin{equation}
\boxed{V=1-\frac{2M}{R}} \quad \text{ and } \quad \boxed{\Sigma=R^2+a^2+\frac{2Ma^2}{R}} \,\,. 
\end{equation}
The determinant is
\begin{equation}
h=k^2 \left(-V \Sigma-\frac{4 a^2 M^2}{R^2}\right)\, ,
\end{equation}
and so
\begin{equation} \label{n^2rot}
n^2 = \frac{h_{00}}{h} = \frac{V -\Sigma \Omega^2+\frac{4aM \Omega}{R}}{ k^2 \left(\Sigma V +\frac{4 a^2M^2}{R^2}\right)}.
\end{equation}
Noting that the velocity of the string with respect to the {\em zero angular momentum observers} (ZAMOs) is
\begin{equation} \label{velocity}
\boxed{v = \frac{\sqrt{\Sigma}\left(\Omega- \frac{2aM}{R \Sigma}\right)}{\sqrt{V+\frac{4 a^2M^2}{R^2 \Sigma}}}}\,\,,
\end{equation}
we can rewrite equation~\eqref{n^2rot} in the suggestive form
\begin{equation} \label{R/R_0}
\boxed{\frac{R}{R_0} = \frac1n \sqrt{1-v^2}}\,\,,
\end{equation}
relating the rotating loop's radius to the unstretched radius by the combined effects of stretching and length contraction.

After a long but straightforward computation, the equations of motion \eqref{motion} boil down to
\begin{equation} \label{Kerr_eq}
\boxed{\frac{p R^2 (M-R)}{a^2+R^2V}=\frac{(p+\rho ) \left(M (a \Omega-1)^2-R^3 \Omega^2\right)}{a^2 \Omega^2+R^2\Omega^2-V+\frac{2 M}{R} \left((a \Omega-1)^2-1 \right)}}\,\,.
\end{equation}
In fact, this is the generalized sail equation \eqref{constraints}; the conservation equations \eqref{conservation2} hold automatically for this embedding. Since from~\eqref{rhop} we have $\rho=\rho(n^2)$ and $p=p(n^2)$, equations~\eqref{velocity}, \eqref{R/R_0} and \eqref{Kerr_eq} relate the four unknown quantities $v$, $\Omega$, $R$ and $n$ ($R_0$ is a known parameter of the loop). Hence, for each value of $R$, say, one can solve these equations to find $v$, $\Omega$ and $n$.

It is interesting to note that the equilibrium configuration for a non-rotating string loop ($\Omega=0$) satisfies
\begin{equation}
\frac{p}{\rho} = \frac{M\left(VR^2 + a^2\right)}{R^3V^2 - Ma^2} .
\end{equation}
In particular, we see that if $a \neq 0$ then $p$ necessarily blows up at some radius $R_* > 2M$ (still outside the ergoregion), and even becomes negative for $R<R_*$, contrary to what the Newtonian intuition would suggest. We explain this bizarre behavior in Appendix~\ref{appendixB}.

\subsection{Radius of the rotating string loop and the speed of sound} 

In Minkowski's spacetime ($M=a=0$), equation~\eqref{Kerr_eq} reduces to (using also equation~\eqref{velocity})
\begin{equation} \label{v^2}
\boxed{v^2 = - \frac{p}{\rho}}\,\,.
\end{equation}
Remarkably, the loop's rotation speed coincides with the transverse speed of sound. We will see later on some interesting consequences of this fact. From \eqref{R/R_0} we then have
\begin{equation}
\frac{R}{R_0} = \frac1n \sqrt{\frac{\rho + p}{\rho}} = \sqrt{\frac{2 F'(n^2)}{F(n^2)}}.
\end{equation}
For a rigid string loop, for instance, we obtain
\begin{equation}
\frac{R}{R_0} = \sqrt{\frac{2}{n^2 + 1}} = \sqrt{1-\frac{p}{\rho}} = \sqrt{1+v^2},
\end{equation}
in agreement with \cite{Brotas68, McCrea71}.

Let us now assume that the string admits a relaxed configuration, satisfies the weak energy condition and possesses a well defined longitudinal speed of sound. Recall that in particular $\rho$ is a strictly increasing function of $n^2$, and $p$ is a nondecreasing function of $n^2$. Suppose that
\begin{equation}
\frac{R}{R_0} \leq 1 \Leftrightarrow 2 F'(n^2) - F(n^2) \leq 0
\end{equation}
for some $v > 0$. From \eqref{v^2} we must have $p < 0$, and hence $n^2 < 1$. Since from \eqref{p(1)=0}
\begin{equation}
2 F'(1) - F(1) = 0,
\end{equation}
we see that there exists $x \in (n^2 , 1)$ such that
\begin{equation}
2 F''(x) - F'(x) \geq 0.
\end{equation}
Therefore
\begin{equation}
2 \frac{F''(x)}{F'(x)} \geq 1 \Rightarrow c^2 \geq x + 1 > 1,
\end{equation}
where $c$ is the longitudinal speed of sound corresponding to the number density $\sqrt{x}$. Thus we have proved the following result:

\begin{Thm}
In Minkowski's spacetime, the radius $R$ of a rotating string loop which admits a relaxed configuration of radius $R_0$, satisfies the weak energy condition and possesses a well defined longitudinal speed of sound not exceeding the speed of light always satisfies $R> R_0$.
\end{Thm}

This is something that one would naively expect due to the balance between the centrifugal and elastic forces, but is not obvious in view of the length contraction term in equation~\eqref{R/R_0}. Indeed, for a string loop with constant longitudinal speed of sound $c > 1$, for example, we have
\begin{equation}
\lim_{n\to 0}\frac{R}{R_0} = \lim_{n \to 0}\sqrt{\frac{(c^2+1)n^{c^2-1}}{n^{c^2+1}+c^2}} = 0.
\end{equation}

None of the above applies if $M>0$, which is again expected: if we have gravity besides the centrifugal and elastic forces then it is possible to have equilibrium configurations even if the string is compressed. For instance, if the string loop is at rest with respect to the ZAMOs, that is, if $\Omega = \frac{2aM}{R \Sigma}$, equation~\eqref{Kerr_eq} yields
\begin{equation}
\frac{p}{\rho} = \frac{M(R^4 + 2a^2R^2V + a^4)}{(R^3 - Ma^2)(R^2V + a^2)} > 0,
\end{equation}
and so $p > 0$ outside the ergoregion (as opposed to the case $\Omega=0$). Since $p$ is a nondecreasing function of $n^2$, we have $n^2 > 1$, and so from \eqref{R/R_0}
\begin{equation}
\frac{R}{R_0} = \frac1n < 1.
\end{equation}
\subsection{Linear Stability}
To analyse the linear stability of rotating string loops we consider the embedding
\begin{equation} \label{embed2}
	\begin{cases}
		t(\tau,\lambda)=\tau \\
		r(\tau,\lambda) = R + \delta r(\tau, \lambda)\\
		\theta(\tau, \lambda)=\frac{\pi}{2}+\delta \theta(\tau, \lambda) \\
		\varphi(\tau,\lambda) = \Omega \tau +  k \lambda + \delta \varphi(\tau,\lambda)
	\end{cases}
\end{equation}
around an equilibrium configuration (satisfying \eqref{Kerr_eq}). To simplify the stability analysis, we consider only the zero angular momentum case $a=0$. We also assume that the string loop satisfies the weak energy condition and has well defined longitudinal and transverse speeds of sound $c$ and $s$ (a necessary condition for stability, as noted in \cite{Carter89b}), with $c,s \leq 1$. Substituting \eqref{embed2} in \eqref{K_metric} we obtain, to first order,

\begin{equation}
 (h_{AB}) = 
 \left(
\begin{matrix}
-V-\frac{2 M}{R^2} \delta r + \Omega^2 (R^2+2 R \, \delta r)+2 R^2 \Omega \, \delta \dot{\varphi} & \, \, R^2 (k\, \Omega + \Omega \, \delta \varphi'+k \, \delta \dot{\varphi})+2 R\, \Omega\, k \, \delta r  \\
R^2 (k\, \Omega + \Omega \, \delta \varphi'+k \, \delta \dot{\varphi})+2 R\, \Omega\, k \, \delta r &  R^2(k^2+2k \, \delta \varphi')+2 R\, k^2 \, \delta r
\end{matrix}
\right) ,
\end{equation}
where once again $\,\,\dot{} \equiv \frac{\partial}{\partial \tau}$ and $\,\,'\equiv \frac{\partial}{\partial \lambda}$.
Therefore we have, to first order,
\begin{equation}
h = - V k^2 R^2-2 (k^2 M \, \delta r+V k^2 R  \, \delta r+V k R^2  \, \delta \varphi')
\end{equation}
and
\begin{equation}
n^2 = \frac{1}{(k R)^2}-\frac{\Omega^2}{V k^2}-2\frac{(V^2 k-k M R \,\Omega^2) \delta r+ (V^2 R -V R^3 \Omega^2 ) \delta \varphi'+V k R^3 \Omega \, \delta \dot{\varphi}}{V^2 k^3 R^3} .
\end{equation}
Substituting in \eqref{TAB} and \eqref{motion}, we obtain the linearized equations of motion
\begin{equation} \label{motion_linear1}
\boxed{A \, \delta r'+B \, \delta \varphi ''+ C \, \delta \dot{r}+ D \, \delta \dot{\varphi}' + E \, \delta \ddot{\varphi}=0} \,\, ,
\end{equation}
\begin{equation} \label{motion_linear2}
\boxed{F \, \delta r + G \, \delta \varphi'+H \, \delta r''+I \, \delta \dot{\varphi}+ J \, \delta \dot{r}'+L \, \delta \ddot{r}=0} \,\, ,
\end{equation}
\begin{equation} \label{motion_linear3}
\boxed{N\, \delta \theta + H\, \delta \theta''+ J\, \delta \dot{\theta}'+ L\, \delta \ddot{\theta} =0} \,\, ,
\end{equation}
with the coefficients given in Appendix~\ref{appendixC}. Note that the polar perturbations $\delta \theta$ are decoupled from the equatorial perturbations $\delta r$ and $\delta \varphi$. This is consistent with the results obtained in \cite{CM93, Martin94} for the case $M=a=0$ (Minkowski geometry).

Since the string loop is closed, we can write the perturbations as  Fourier series:
\begin{equation}\label{fourier1}
\delta r(\tau, \lambda)=\sum_{j=-\infty}^{+\infty}c_j(\tau) e^{i j k \lambda} ,
\end{equation}
\begin{equation}\label{fourier2}
\delta \varphi(\tau, \lambda)=\sum_{j=-\infty}^{+\infty}d_j(\tau) e^{i j k \lambda} ,
\end{equation}
\begin{equation}\label{fourier3}
\delta \theta(\tau, \lambda)=\sum_{j=-\infty}^{+\infty}e_j(\tau) e^{i j k \lambda} .
\end{equation}
Substituting \eqref{fourier1} and \eqref{fourier2} in equations \eqref{motion_linear1} and \eqref{motion_linear2}, we obtain for the equatorial perturbations
\begin{equation}
\frac{d}{d \tau} \left( \begin{array}{c}
c_j \\
\dot{c}_j  \\
d_j \\
\dot{d}_j \end{array} \right)=\textbf{M}_j \left( \begin{array}{c}
c_j \\
\dot{c}_j  \\
d_j \\
\dot{d}_j \end{array} \right) ,
\end{equation}
with 
\begin{equation}
\textbf{M}_j=\left( \begin{array}{c c c c}
0 & 1 & 0 & 0 \\
-\frac{F}{L}+\frac{H}{L} j^2 k^2 & -i \frac{J}{L} j k & -i \frac{G}{L} j k & -\frac{I}{L} \\
0 & 0 & 0 & 1 \\
-i \frac{A}{E} j k & -\frac{C}{E} & \frac{B}{E} j^2 k^2 & -i \frac{D}{E} j k \end{array} \right).
\end{equation}
Looking for solutions proportional to $e^{i \omega \tau}$, we are led to the characteristic polynomial
\begin{empheq}[box=\widefbox]{align} 
\begin{split} \label{polynomialequ}
\tilde{p}_j(\omega)&=\text{det}(\textbf{M}_j-i \omega \textbf{I}) \\
&= \frac{1}{E L}\left[ k^2 j^2 ( A G -B F + B H k^2 j^2) + 
k j \left( C G + A I -D F + (D H + B J) k^2 j^2 \right) \omega  \right. \\ &\left. +\left(  C I -E F + 
k^2 (E H + D J + B L) j^2 \right) \omega^2  + 
k (E J + D L) j \, \omega^3 + E L\, \omega^4\right] ,
\end{split}
\end{empheq}
where $\textbf{I}$ represents the $4\times4$ identity matrix. Note that $\tilde{p}_j(\omega)$ is a quartic function of $\omega$ with real coefficients, and so the complex conjugate of a root of $\tilde{p}_j$ is also a root.

For polar perturbations, substituting \eqref{fourier3} in equation \eqref{motion_linear3}, we obtain
\begin{equation}
L\, \ddot{e}_j+i j k J\, \dot{e}_j+\left(N - j^2 k^2 H\right) e_j= 0  .
\end{equation}
Again looking for solutions proportional to $e^{i \omega \tau}$, we obtain the characteristic polynomial
\begin{equation} \label{polynomialpol}
\boxed{\tilde{q}_j(\omega)=L\, \omega^2+j k J\, \omega -\left( N -j^2 k^2 H \right)} \,\, .
\end{equation}
Once more, since $\tilde{q}_j(\omega)$ is a quadratic function of $\omega$ with real coefficients, the complex conjugate of each of its roots is also a root. 

By the theory of linear ODEs, a necessary condition for the embedding \eqref{embed} to be linearly stable under equatorial (polar) perturbations is that, for each mode $j \in \bbZ$, all the roots of $\tilde{p}_j$ ($\tilde{q}_j$) have non-negative imaginary parts. Since these roots come in conjugate pairs, this necessary condition amounts to requiring all roots to be real. Under such an assumption, linear stability is equivalent to the absence of secular terms, that is, to all roots having the same algebraic and geometric multiplicities (which is automatically true for simple roots).

\subsubsection{Minkowski spacetime}

In the Minkowski case $M=a=0$, the roots of $\tilde{p}_j$ are the solutions of 
\begin{align} \label{pMink}
\begin{split}
	&\left(\omega -j \frac{s}{R}  \right) \bigg[ R^3 (1+s^2)(1-c^2 s^2)\, \omega^3 +j R^2 s\, (s^2-1) \big[c^2 (2+3 s^2)-1 \big]\, \omega^2 \\ &+ R \Big((s^2-3) s^2 -\big[1-3 s^2 +j^2(1-s^2)^2(1+3 s^2)\big]c^2 \Big)\, \omega - j s(1-s^2) \Big[c^2\Big(j^2 (1-s^2)^2 -3 \Big)-s^2 \Big] \bigg]=0 .
	\end{split} 
\end{align}  
We can immediately see that 
\begin{align}
	\omega=j \frac{s}{R}=j\, \Omega
\end{align}
are real roots of $\tilde{p}_j$. These roots correspond to a special kind of solutions, namely transverse waves travelling along the string loop in the opposite direction to its rotation \cite{CM93, Martin94}.\footnote{These solutions correspond to a fixed (in time) equatorial deformation of the rotating string loop. They exist because in Minkowski spacetime $v=s$, and so we can have transverse (equatorial) perturbations moving with respect to the string loop with exactly minus the velocity of the rotating string loop. Explicitly, if $\omega=j\, \Omega$ for each mode $j\in\bbZ$ we have
$$\delta r(\tau, \lambda)=\sum_{j=-\infty}^{+\infty}c_j e^{i j \Omega \tau} e^{i j k \lambda}=\sum_{j=-\infty}^{+\infty}c_j e^{i j(\Omega \tau+k \lambda )}= \sum_{j=-\infty}^{+\infty}c_j e^{i j \varphi}=\delta r(\varphi) ,$$ 
and similarly for $\delta \varphi$. We can apply the same argument to predict the existence of this kind of solutions for polar perturbations, and, in fact, it is simple to show their existence explicitly.} Now, since there is one known root of $\tilde{p}_j$ for each mode $j\in\bbZ$, we only need to analyse the roots of the remaining cubic polynomial in \eqref{pMink}.
To analyse the roots of this polynomial, we make use of its discriminant:
\begin{align}
	\begin{split} \label{discriminantMink}
	&\Delta=4 R^6\bigg( c^2 (c^2- s^2)^2 (1-s^2)^6\, j^6 \\ &- (1-s^2)^4 \left[c^8 s^2- c^6 (3+28 s^2+ 21 s^4)-c^4 s^2 (41+72 s^2+41 s^4) -c^2 s^4(21+28 s^2+3 s^4)+ s^6 \right] \, j^4\\
	&- (1-s^2)^2  \Big[2 c^8 s^2(1+12 s^2+9 s^4) -c^6 (3+20 s^2+78 s^4+4 s^6 -57 s^8) +8 c^4 s^2 (7-13 s^4 +7 s^8) \\ &+c^2 s^4 (57-4 s^2 -78 s^4 -20 s^6-3 s^8) +2 s^6(9+12 s^2+s^4)   \Big] j^2 \\&- (1+s^2) (1- s^2 c^2)\Big[s^4-3 s^2 (1-c^2)-c^2  \Big]^3 \bigg) .
	\end{split} 
\end{align}
If $\Delta>0$ for a given mode $j\in\bbZ$ then the cubic polynomial has three distinct real roots; if $\Delta=0$ then the polynomial has a multiple root (and all the roots are real); and if $\Delta<0$ then the polynomial has one real root and a pair of non-real (conjugate) roots. Therefore, by the discussion of the previous section, if $\Delta< 0$ for some mode $j\in\bbZ$ then the equilibrium is linearly unstable under equatorial perturbations. On the other hand, if $\Delta>0$ for every mode $j\in\bbZ$, with all the roots different from $j (s/R)$, then the equilibrium is linearly stable under equatorial perturbations.
 
 Setting 
 \begin{equation}
 	y=j^2\ge 0 ,
 \end{equation}
 we see that $\Delta$ is a cubic function of $y$ for $s\neq c$, with the coefficients being functions of $s$ and $c$. In particular, the coefficient of $y^3$ is positive. Therefore we can rewrite the expression of the discriminant as 
 
 \begin{equation}
 	\Delta(y)=\widetilde{A}y^3+\widetilde{B}y^2+ \widetilde{C} y+ \widetilde{D}  ,
 \end{equation}
 with $\widetilde{A}>0$, $\widetilde{B}$, $\widetilde{C}$ and $\widetilde{D}$ the obvious functions of $s$ and $c$ in \eqref{discriminantMink}.
 If the inequality
 \begin{equation} \label{ineq}
 	\widetilde{B}^2-3 \widetilde{A}\widetilde{C}\leq 0
 \end{equation}
 holds, then the global minimum of $\Delta(y)$ for $y\geq0$ is $\Delta(0)=\widetilde{D}$. On the other hand, if the inequality \eqref{ineq} does not hold, then there exist $y_1$ and $y_2$, with $y_1<y_2$, such that $\Delta$ has a local maximum at $y=y_1$ and a local minimum at $y=y_2$. In this case, if
  \begin{equation}
 \widetilde{C}\leq0 ,
 \end{equation}
 then the global minimum of $\Delta(y)$ for $y\geq0$ is $\Delta(y_2)$. On the other hand, if
 \begin{equation}
 	\widetilde{C}>0 \quad \text{ and } \quad \widetilde{B}\geq 0 ,
 \end{equation}
then the global minimum of $\Delta(y)$ for non-negative $y$ is $\Delta(0)=\widetilde{D}$. If instead
\begin{equation}
 	\widetilde{C}>0 \quad \text{ and } \quad \widetilde{B}< 0 ,
\end{equation}
then the global minimum of $\Delta(y)$ for non-negative $y$ is $\min \left\{\widetilde{D}, \Delta(y_2) \right\}$.

To summarize, for given $s$ and $c$ we know the global minimum of the discriminant $\Delta(y)$ for $y \geq 0$. If this global minimum is positive then the equilibrium is linearly stable under equatorial perturbations. Therefore, as shown in Figure~\ref{stabilitymink}, the equilibrium is linearly stable for $s<c$. If $s=c$, it is easy to see that $\widetilde{A}=0$, $\widetilde{B}>0$ and $\widetilde{C}<0$. Therefore, there exists $y_3$ such that $\Delta(y)$ has a local minimum at $y=y_3>0$, and so the global minimum of $\Delta(y)$ for non-negative $y$ is $\Delta(y_3)$. It is easy to show that $\Delta(y_3)=0$, and so the equilibrium is also linearly stable under equatorial perturbations in this case.\footnote{In particular, the ``warm'' cosmic string model in Section~\ref{warm} yields linearly stable rotating loops.}

\begin{figure}[h] 
	\centering 	
	\includegraphics[scale=0.5]{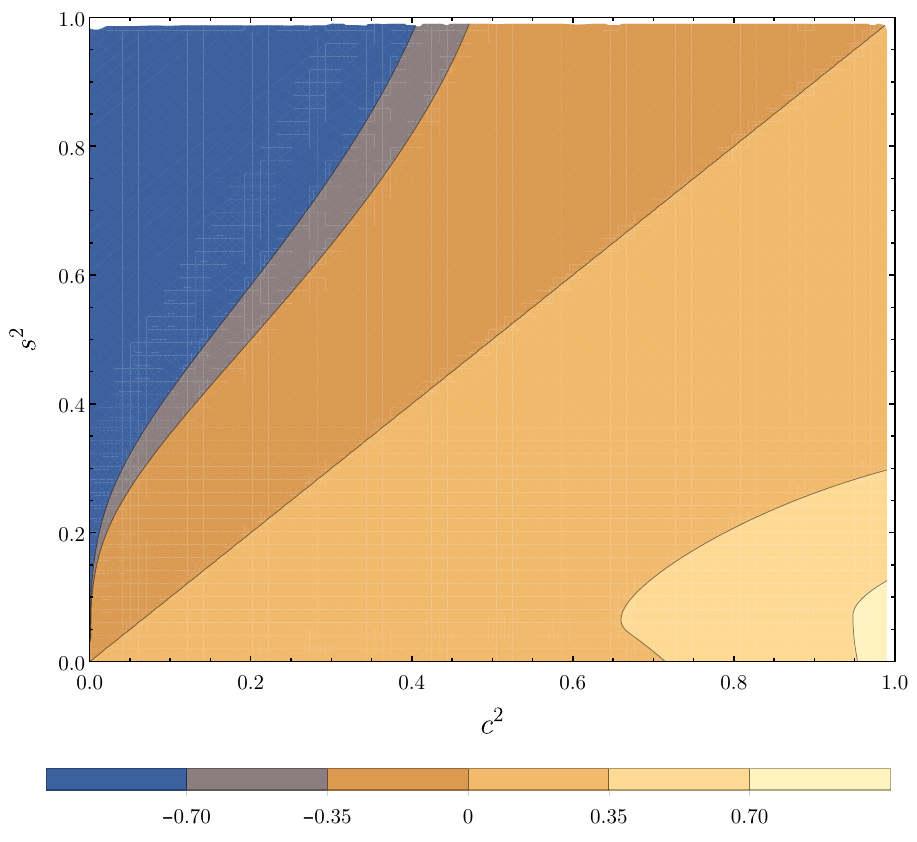} 
	\caption{Contour plot of the global minimum of $\text{tanh}\left[\Delta/(4 R^6) \right]$ as a function of $c^2$ and $s^2$.}
	\label{stabilitymink}
\end{figure}

If the global minimum is negative, however, we cannot conclude immediately that the equilibrium is linearly unstable under equatorial perturbations, because this minimum need not occur near a perfect square $y=j^2$ (with $j \in \bbZ$). Instead, we have to look at the values of the discriminant on the two perfect squares around the point of global minimum, and check whether the smallest of these values is negative. This is shown in Figure~\ref{stabilitymink2}, where we see that indeed the region of stability persists beyond $s=c$, the first instabilities ocurring approximately for 
\begin{equation}
s^2=c^2 + \frac{12 c^2 (1-c^2)}{15 c^2+4}.
\end{equation}

\begin{figure}[h] 
	\centering 	
	\includegraphics[scale=0.5]{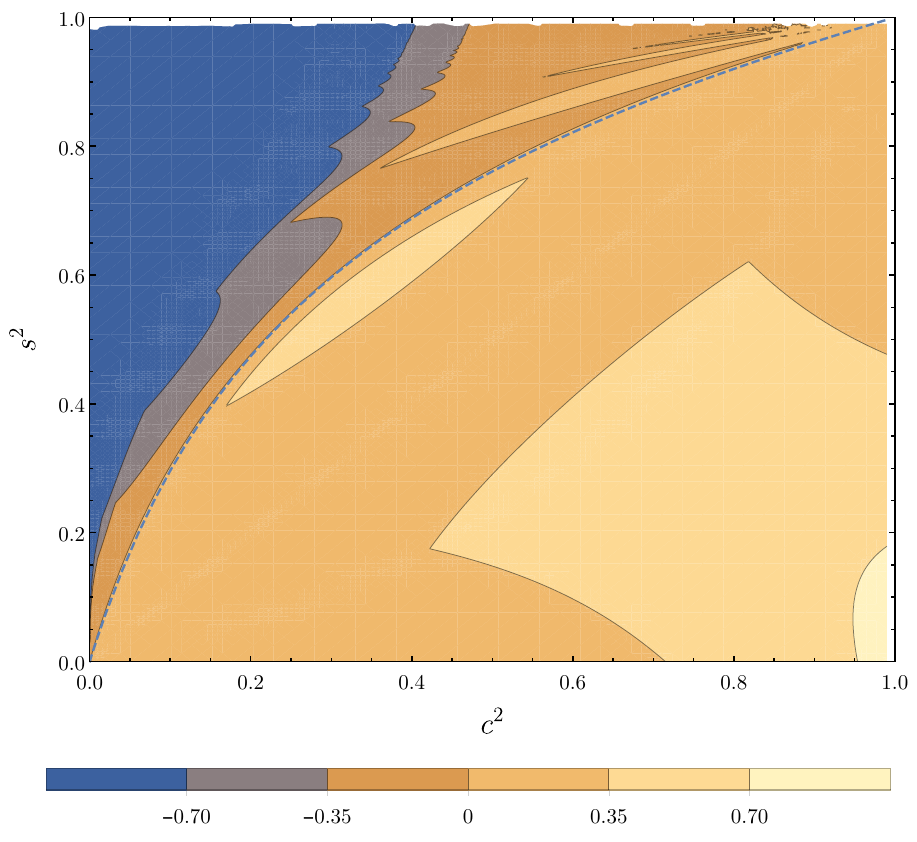} 
	\caption{Contour plot of the minimum of $\text{tanh}\left[\Delta/(4 R^6) \right]$ at a perfect square as a function of $c^2$ and $s^2$. The plot of $s^2=c^2 + 12 c^2 (1-c^2) / (15 c^2+4)$ is shown as dashed line. }
	\label{stabilitymink2}
\end{figure}

To be completely rigorous in establishing linear stability, we must check that the distinct real roots of the cubic polynomial in \eqref{pMink} do not coincide with any of the special roots $\omega=j \frac{s}{R}=j\, \Omega$. Evaluating the  polynomial at these roots yields
\begin{equation}
s (s^2-c^2) j(j^2-1),
\end{equation}
and so the roots are simple unless $s=0$, $s=c$, $j=0$ or $j=\pm 1$. We will not analyse the case $s=0$ in detail because it is unstable under polar perturbations, as we will see. In the case $s=c$, with $s>0$, one can check that the geometric multiplicity of the roots is two, and so we have stability. The geometric multiplicity of the roots with $j=0$, however, is one, and so we have secular terms; but these do not signal any instability, because they connect to infinitesimally close equilibria with slightly different radius, angular velocity and initial orientation of the string loop (that is, $\delta r = \alpha$ and $\delta \varphi =\beta t + \gamma$ for appropriate constants $\alpha$, $\beta$ and $\gamma$). Finally, the roots with $j=\pm 1$ also have geometric multiplicity one, and again lead to secular terms; these connect to equilibrium solutions obtained by slightly boosting the equilibrium on the equatorial plane, and are an unavoidable consequence of the Lorentz invariance of Minkowski space.

On the other hand, the polynomial $\tilde{q}_j$ is given in the Minkowski case $M=a=0$ by
\begin{equation}
\tilde{q}_j(\omega)=(1+s^2)\, \omega^2- \frac{2js^3}{R}\, \omega - \frac{j^2s^2(1-s^2)}{R^2} ,
\end{equation}
with discriminant
\begin{equation}
	\Delta_{\tilde{q}_j}=4 j^2 \frac{s^2}{R^2} .
\end{equation}
Since this discriminant is always nonnegative, the polynomial has real roots, which are simple unless $s=0$ or $j=0$.\footnote{One of these roots is, of course, $\omega=j\frac{s}{R}$.} If $s=0$ there are secular terms, and the equilibrium is unstable, as promised. When $j=0$ there are also secular terms, but these connect to equilibrium solutions obtained by slightly boosting the equilibrium in the $z$-direction, and are again required by Lorentz invariance. Thus, the equilibrium is linearly stable under polar perturbations for $s>0$.

Noticing that $s=v$ in Minkowski's spacetime, we have proved the following result:
\begin{Thm}
In Minkowski's spacetime, the equilibrium of a rotating string loop which satisfies the weak energy condition and has well defined longitudinal and transverse speeds of sound $c \leq 1$ and $s \leq 1$ is linearly stable if its rotation speed $v=s$ satisfies $0 < v\leq s(c)$, where the function $s(c)$ is approximately given by $s^2(c) = c^2 + 12 c^2 (1-c^2)/(15 c^2+4)$, and linearly unstable if $v=0$ or for certain values of $v>s(c)$. Moreover, the equilibrium is linearly stable under polar perturbations for any $v>0$. 
\end{Thm}

These results are consistent with those obtained in \cite{CM93, Martin94}, where a different parameterization of the worldsheet, using $\varphi$ instead of $\lambda$, was employed. In particular, the instability region for the first 400 modes plotted in~\cite{Martin94} is a good approximation of the full instability region, depicted in Figure~\ref{stabilitymink2}, although around the point $(c^2,s^2)=(0,1)$ the instability occurs for modes with $j>400$.
\subsubsection{Schwarzschild spacetime}
In the Schwarzschild spacetime, the speed of the rotating string loop is given by
\begin{equation}
	v^2=\frac{(R-2 M) s^2+M}{(R-2 M)+M s^2} .
\end{equation}
Thus, we see that $v^2<1$ implies $R>3 M$. This is related to the fact that all circular orbits by massive particles must have a larger radius than the radius $3 M$ of the photon sphere: we must have $R>3M$ for an incoherent loop (dust ring), and any elastic force will just add to the gravitational force, since we are assuming that the transverse speed of sound is well defined (hence $p\leq 0$).

Generating contour plots of the discriminant of $\tilde{p}_1$ as a function of $c^2$ and $s^2$, for different values of $R$, we can see that this discriminant is always negative. In Figs.~\ref{fig:discriminantschw4}, \ref{fig:discriminantschw7} and \ref{fig:discriminantschw11} we present some of these plots. Since the discriminant of $\tilde{p}_1$ is negative, the polynomial has two complex conjugate non-real roots. Thus, we see that the equilibrium is linearly unstable under equatorial perturbations.\footnote{This is related to the well known Ringworld instability, discussed in~\cite{Breakwell81, McInnes03}.}


\begin{figure}[h!] 
	\centering 	
	\includegraphics[scale=0.43]{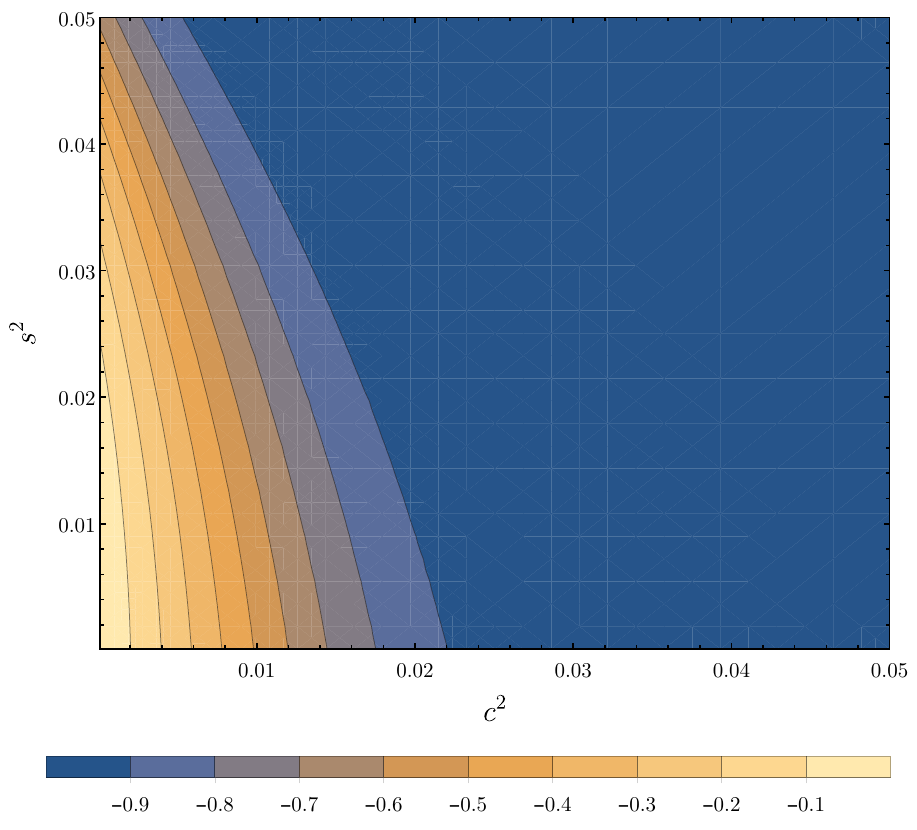} 
	\caption{Contour plot of $\text{tanh}\left(\Delta_{\tilde{p}_1}\right)$ as a function of $c^2$ and $s^2$ with $R=4 M$.}
	\label{fig:discriminantschw4}
\end{figure}

\begin{figure}[h!]
	\centering 	
	\includegraphics[scale=0.43]{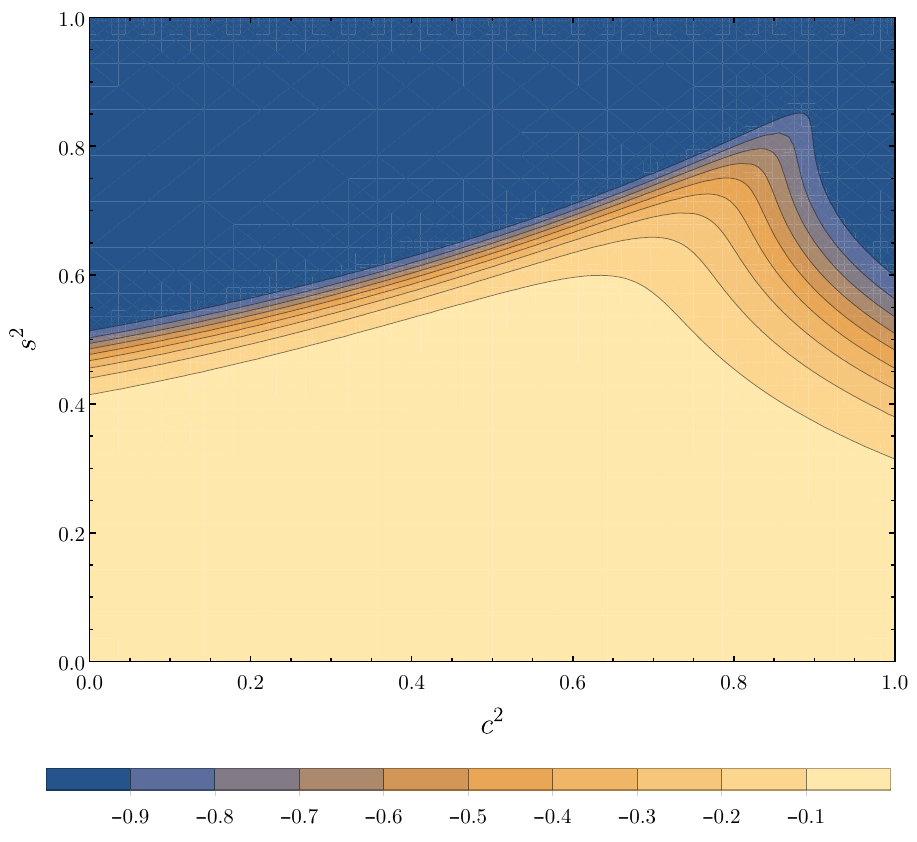} 
	\caption{Contour plot of $\text{tanh}\left(\Delta_{\tilde{p}_1}\right)$ as a function of $c^2$ and $s^2$ with $R=7 M$.}
	 \label{fig:discriminantschw7}
\end{figure}

\begin{figure}[h!] 
	\centering 	
	\includegraphics[scale=0.43]{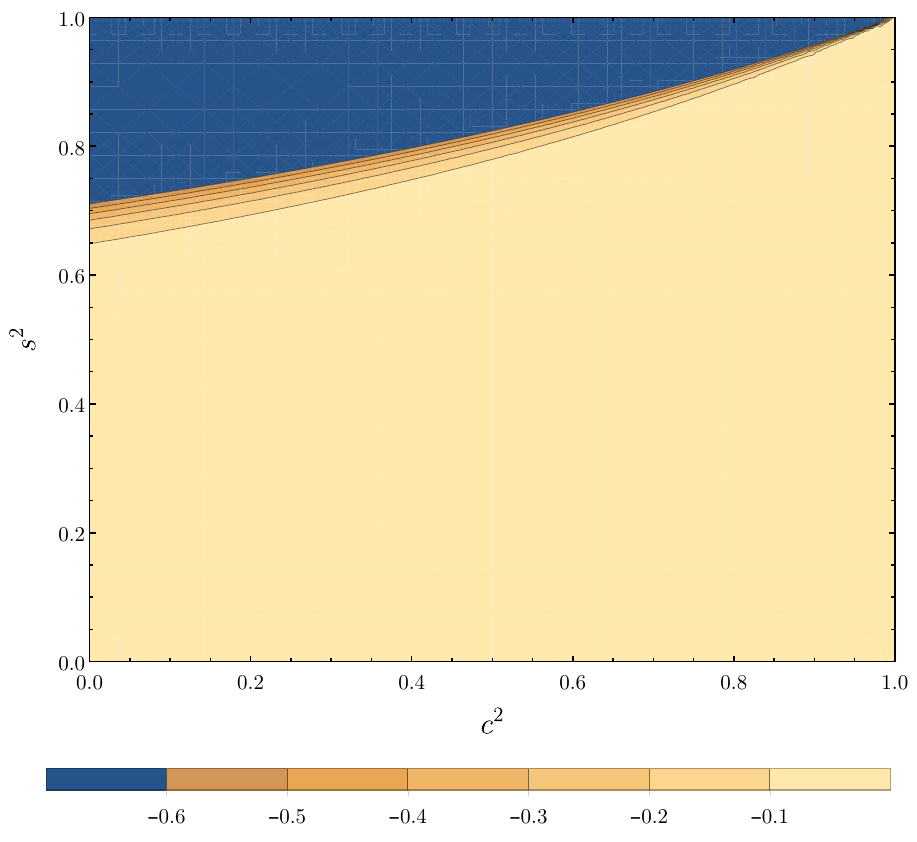} 
	\caption{Contour plot of $\text{tanh}\left(\Delta_{\tilde{p}_1}\right)$ as a function of $c^2$ and $s^2$ with $R=11 M$.}
	\label{fig:discriminantschw11}
\end{figure}


On the other hand, the discriminant of the polynomial $\tilde{q}$ is
\begin{equation}
	\Delta_{\tilde{q}_j}=\frac{4}{R} \left[M \left(\frac{1+s^2}{R-3 M}\right)^2 +\frac{s^2}{R-2 M} j^2 \right].
\end{equation}
Since this discriminant is positive, the polynomial has two distinct real roots.
Thus the equilibrium is linearly stable under polar perturbations.  
The results obtained in this section can be summarized as follows:
\begin{Thm}
In Schwarzschild's spacetime, the equilibrium of a rotating string loop which satisfies the weak energy condition and has well defined longitudinal and transverse speeds of sound $c \leq 1$ and $s \leq 1$ 
is linearly unstable under equatorial perturbations, and linearly stable under polar perturbations. 
\end{Thm}

\section{Gedanken experiments with rotating string loops on a Kerr background}\label{section3}

In \cite{W74}, Wald initiated a series of gedanken experiments to test the validity of the cosmic censorship conjecture. These consisted of using test particles to try to destroy the event horizon of an extremal black hole in order to undress the singularity. Around the same time, in \cite{Penrose71}, Penrose devised a mechanism for extracting energy from Kerr black holes: since the energy can be negative in the ergoregion, we can let a test particle (of positive energy) fall in from infinity and break up into two inside the ergoregion; if one of the pieces has negative energy, the other piece must have more energy than the original particle. Provided that this more energetic particle can reach infinity, we have \emph{extracted} energy from the Kerr black hole.

The gedanken experiments described above use point particles, which are the simplest kind of matter model one can think of, and we may wonder what happens if we consider extended objects instead. In this section, we extend Wald's and Penrose's gedanken experiments to rotating string loops. On the one hand, we try to overspin an extremal Kerr black hole by letting it absorb a string loop. On the other hand, we try to increase the efficiency of the Penrose process by using elastic string loops.  

The motion of Nambu-Goto and cosmic strings in black hole spacetimes has previously been studied in \cite{Larsen92, Larsen94, FL99, DVF99, Suzuki08, JS09b, Suzuki10, II10, II10a}. Other mechanisms for extracting energy from Kerr black holes using strings were considered in \cite{KP00, KIT16}.

\subsection{Elastic string loops cannot destroy extremal black holes}
We want to try to spin up a Kerr black hole past extremality through the absorption of rotating elastic string loops described by an arbitrary elastic law. We will consider the extremal Kerr solution by setting $a=M$, and use, instead of the usual Boyer-Lindquist form of the metric \eqref{K_metric}, Painlev\'{e}-Gullstrand coordinates \cite{Natario09}, so that the metric for the equatorial plane $\theta=\frac{\pi}{2}$ reads
\begin{equation} \label{Kerr_metric}
ds^2=-dt^2+\frac{r ^2}{\Sigma}(dr-vdt)^2+ \Sigma(d\varphi - \omega dt)^2\,,
\end{equation}
where
\begin{equation}
\Sigma = r^2 + a^2 + \frac{2Ma^2}{r} \; ;
\end{equation}
\begin{equation}
\omega = \frac{2Ma}{r \Sigma} \; ;
\end{equation}
\begin{equation}
v = -\frac{\sqrt{2Mr(r^2+a^2)}}{r^2}\,.
\end{equation}
Note that $\omega$ and $v$ are just the familiar expressions for the angular velocity and radial proper velocity of a zero angular momentum observer dropped from infinity. In our analysis, we will need to compute the energy and the angular momentum of the string loop at the horizon; this is why this coordinate system, which is well defined at the horizon, is appropriate. Moreover, we will choose a different embedding from the the previous sections, since we no longer want to be restricted to equilibrium configurations:
\begin{equation} \label{embed_Kerr}
\begin{cases}
t(\tau,\lambda)=\tau \\
r(\tau,\lambda) = R(\tau) \\
\varphi(\tau,\lambda) = \Phi(\tau) +  k \lambda
\end{cases}\,,
\end{equation}
where $k$ and $\lambda$ are defined as before. This embedding represents a string loop which may move freely along the radial and angular directions but still preserves its axial symmetry. Substituting~\eqref{embed_Kerr} in \eqref{Kerr_metric} yields
\begin{equation}
(h_{AB}) = 
\left(
\begin{matrix}
-1 + \frac{R ^2}{\Sigma}{v_{R}}^2+\Sigma {v_{\Phi}}^2 & k \Sigma v_{\Phi} \\
k \Sigma v_{\Phi} & \Sigma k^2 
\end{matrix}
\right),
\end{equation}
where
\begin{equation} 
\begin{cases}
v_{R}=\dot{R}-v \\
v_{\Phi}=\dot{\Phi} - \omega \\
\end{cases} ,
\end{equation}
and so
\begin{equation} 
h=-k^2(\Sigma - R ^2 {v_{R}}^2)\, .
\end{equation}
The Kerr metric admits two Killing vector fields, $\frac{\partial}{\partial t}$ and $\frac{\partial}{\partial \varphi}$, whose associated covector fields are given on the worldsheet by
\begin{equation}\label{covectors}
\begin{cases}
\xi ^{t}=g\left(\frac{\partial}{\partial t}, \cdot\right)=\left(-1-\frac{R ^2}{\Sigma}vv_{R}-\Sigma \omega v_{\Phi}\right)d\tau -\left(\Sigma \omega k\right) d\lambda \\
\xi ^{\varphi}=g\left(\frac{\partial}{\partial \varphi}, \cdot\right)=\left(\Sigma v_{\Phi}\right)d\tau + \left(\Sigma k\right) d\lambda \\
\end{cases}.
\end{equation}
We are now able to compute the energy and the angular momentum of a string loop using \eqref{conserved_quantity} and \eqref{covectors}. We obtain
\begin{equation}\label{E_L}
\begin{cases}
E=-\frac{2 \pi \sqrt{-h}}{k}\left[\left(\left(\Sigma + R^2 v v_{R}\right)\left(\frac{1}{\Sigma h_{00}}-\frac{k^2}{h}\right)+\frac{\Sigma  \omega v_{\Phi}}{h_{00}} \right)p+\left( \frac{1}{h_{00}}+\frac{R^2 v v_{R}}{\Sigma h_{00}} + \frac{\Sigma \omega v_{\Phi}}{h_{00}}\right)\rho \right] \\
L=-\frac{2\pi\sqrt{-h}\Sigma v_{\Phi}}{kh_{00}}(p+\rho)
\\
\end{cases}.
\end{equation}
Now, assume that an elastic string loop with infinitesimal energy $E$ and angular momentum $L$ enters an extremal black hole. Ignoring the backreaction of the string loop on the spacetime, we can easily find a relationship between $E$ and $L$ that must be satisfied in order for the black hole to be destroyed: if $a_f$ and $M_f$ are the parameters of the black hole after absorbing the string loop then
\begin{equation}
\frac{a_f}{M_f}=\frac{Ma+L}{(M+E)^2} \approx \frac{a}{M}+\frac{E}{M}\left(\frac{L}{EM}-\frac{2a}{M}\right)\, .
\end{equation}
At extremality $a=M$, and the black hole is destroyed if $a_f>M_f$. Hence, a string loop capable of destroying the black hole must satisfy
\begin{equation} \label{destroy}
\frac{L}{EM}>2\, .
\end{equation}
From \eqref{E_L}, we can write the energy of the string loop in the following way:
\begin{equation}
E=-\frac{2 \pi \sqrt{-h}}{k}\cdot \frac{\left(R^2vv_{R} +\Sigma \right)\left(2F'h+(Fh-2F'h_{00})k^2\Sigma \right)}{h^2 \Sigma}+ \omega L \,.
\end{equation}
In order to simplify our expression, we choose units in which the mass of the black hole is $M=1$. Note that with this choice \eqref{destroy} reduces to
\begin{equation} \label{destroy_unit}
\boxed{\frac{L}{2}>E}\,\, .
\end{equation}
Since we want to restrict our attention to the string loops that actually traverse the horizon, we evaluate the energy at the horizon, which then takes the following form:
\begin{equation}\label{energy_1}
E= F'\frac{2+8\left(-1+\dot{\Phi}\right){\dot{\Phi}}}{k^2\sqrt{-\dot{R}}\left(4+\dot{R}\right)^{\frac{3}{2}}} + F\frac{\sqrt{-\dot{R}}}{\sqrt{4+\dot{R}}}+\frac{L}{2}\,.
\end{equation}
Note that we can do this because our coordinate system is well defined at the horizon. 
All the terms appearing in expression for the energy are implicitly being evaluated at $R=M=1$.
By \eqref{E_L}, we can further see that $\dot{\Phi}$ is given by
\begin{equation}
\dot{\Phi}=\omega+\frac{\sqrt{-h}kL}{4\pi\Sigma F'}\, ,
\end{equation}
and replacing in \eqref{energy_1} yields
\begin{equation} \label{energy_final}
\boxed{E=-\frac{\left(\frac{L^2k^2}{8\pi F'}+4\pi F\right)}{\sqrt{-\dot{R}(4+\dot{R})}}\dot{R}+\frac{L}{2}}\,\, .
\end{equation}
The condition that the worldsheet is timelike requires  
\begin{equation}
-1+\frac{R ^2}{\Sigma}(\dot{R}-v)^2+ \Sigma(\dot{\Phi} - \omega)^2<0\, ,
\end{equation}
and consequently, at the horizon,
\begin{equation}\label{restriction}
\boxed{-4<\dot{R}<0}\,.
\end{equation}
This restriction guarantees that \eqref{energy_final} is well defined. From \eqref{destroy_unit}, \eqref{energy_final} and \eqref{restriction}, in order to destroy the black hole we must have
\begin{equation}
\boxed{\frac{L^2k^2}{8\pi F'}+4\pi F < 0}\,\,.
\end{equation}
This inequality cannot hold if $F' \geq 0$ and $F \geq 0$, that is, if the string obeys the weak energy condition in the worldsheet at the points where it crosses the horizon.
Interestingly, the weak energy condition in the worldsheet can be interpreted as the null energy condition in spacetime.
\begin{Thm}\label{Thm_Equivalent}
The weak energy condition in the worldsheet is equivalent to the null energy condition in spacetime.
\end{Thm}
\begin{proof}
The null energy condition in spacetime states that
\begin{equation}\label{null_space}
T_{\mu \nu} k^\mu k^\nu \geq 0 
\end{equation}
for every null vector $k$ in spacetime. On the other hand, the weak energy condition in the worldsheet states that 
\begin{equation}\label{weak_string}
T_{A B} u^A u^B \geq 0
\end{equation}
for every causal vector $u$ in the worldsheet. We can decompose any null vector in spacetime into a component tangent to the worldsheet and an orthogonal component, from which we can write
\begin{equation} 
T_{\mu \nu} k^\mu k^\nu=T_{A B} k^A k^B + T_{i j} k^i k^j=T_{A B} k^A k^B \,,
\end{equation}
as $T_{ij}=0$. Furthermore, because $k$ is null,
\begin{equation}
0 = k_\mu k^\mu = k_A k^A + k_i k^i \,.
\end{equation}
Since any vector orthogonal to the worldsheet is spacelike, we have $ k_i k^i \geq 0$, and so $ k_A k^A \leq 0$, that is, the component of $k$ tangent to the worldsheet is causal. It follows that \eqref{null_space} and \eqref{weak_string} are equivalent statements.
\end{proof}
Therefore, we can summarize our result as follows:
\begin{Thm}\label{Thm_Safety}
Elastic string loops satisfying the null energy condition at the event horizon cannot destroy an extremal Kerr black hole.
\end{Thm}
This result is consistent with the more general statement, proved by the authors in \cite{NQV16}, that test fields satisfying the null energy condition at the event horizon cannot destroy an extremal Kerr black hole. In fact, the authors conjectured that this is true for any type of localized matter, including test particles and elastic strings. One can therefore be confident that as long as the matter satisfies the null energy condition at the horizon, extremal black holes are safe from destruction. More recently, it was shown in \cite{SW17} that the same statement holds for near-extremal black holes. On the other hand, it is possible to destroy extremal black holes with matter that does not satisfy standard energy conditions \cite{Toth15}. In fact, the null energy condition seems to be a fundamental assumption for the weak cosmic censorship conjecture and the safety of black holes. 

\subsection{Elastic Penrose Process}

In this section we revert to the Boyer-Lindquist form \eqref{K_metric} of the metric, which in the equatorial plane $\theta=\frac{\pi}{2}$ reads
\begin{equation}
ds^2 = - \left( 1 - \frac{2M}{r} \right) dt^2 + \frac{r^2}{\Delta} dr^2 + \left(r^2 + a^2 + \frac{2Ma^2}{r}\right) d\varphi^2 - \frac{4Ma}{r} dt d\varphi.
\end{equation}
This change is simply to facilitate the comparison with the standard particle Penrose process, which is usually analyzed in these coordinates. The expressions for the energy $E$ and angular momentum $L$ of a string loop in Boyer-Lindquist coordinates can be computed similarly to what was done in the previous section for Painlev\'{e}-Gullstrand coordinates:
\begin{equation}
\begin{cases}
E= \frac{2 \pi}{k}\sqrt{-h}\left[\left(1-\frac{R^4 \Sigma \dot{R}^2}{\Delta \left(4a^2M^2-2MR\Sigma+R^2\Sigma\right)}\right)^{-1}p+\frac{\left(1-\frac{2M}{R}+\frac{2aM\dot{\Phi}}{R}\right)}{h_{00}}(p+\rho)\right]\\
L=-\frac{2 \pi}{k}\frac{\sqrt{-h}}{h_{00}}\left(\Sigma \dot{\Phi}-\frac{2aM}{R}\right)(p+\rho)
\end{cases}.
\end{equation}

To study the Penrose process for point particles we we need to consider the  local conservation of the $4$-momentum in the breakup proccess. Usually this accounts for the conservation of energy, angular momentum and the radial component of momentum for motions in the equatorial plane. When working with elastic string loops, however, we don't have a $4$-momentum vector; instead, we have an energy-momentum tensor supported on the worldsheet, which will bifurcate whenever the string loop breaks up into two. As is well known, the energy-momentum tensor itself does not have to be continuous across a shock, but only its components orthogonal to the shock, as required by the divergence theorem; since in the case of a string loop breakup the shock is a one-dimensional curve, we must require that all components of the energy-momentum tensor orthogonal to this curve (summed over the various branches of the worldsheet) remain continuous.\footnote{The same reasoning can be applied to the case of the breakup of a particle, where the shock is now a single point.} In other words, assuming that the worldsheet $\Sigma_0$ breaks up into worldsheets $\Sigma_1$ and $\Sigma_2$ along the curve $\gamma$, we must require that
\begin{equation} \label{conservation_break}
\left(T^{AB} V_A \nu_B\right)_0 = \left(T^{AB} V_A \nu_B\right)_1 + \left(T^{AB} V_A \nu_B\right)_2
\end{equation}
at each point of $\gamma$, where $V$ is an arbitrary spacetime covector and $\nu$ is the past-pointing (say) normal to $\gamma$ in each branch of the worldsheet.


If we consider axially symmetric string loops, conditions \eqref{conservation_break} are invariant under rotations, and we can define an effective $4$-momentum vector for the string loop by integrating the appropriate components of the energy-momentum tensor along the string loop. Indeed, we already have a good definition for two of its components: the energy $E$ and the angular momentum $L$. Since we are restricting the motion of the string loop to the equatorial plane, there's only one component left. We define it by using \eqref{conserved_quantity} with the covector field associated to $\frac{\partial}{\partial r}$, obtaining
\begin{equation}
P_r=\frac{2 \pi}{k}\frac{\sqrt{-h}}{\Delta}R^2\dot{R}\left[\left(\frac{2M}{R}-\frac{4a^2M^2}{R^2}+\Sigma\left(R^2\dot{R}^2-1\right)\right)^{-1}\Sigma p-\frac{1}{h_{00}}(p+\rho) \right].
\end{equation}
It's easy to see that this quantity reduces to the \emph{radial momentum} of test particles by considering a ring of dust, corresponding to $p=0$. Finally, we can define the effective $4$-momentum of the string loop as\footnote{One can think of this covector as being defined on the one-dimensional submanifold of the worldsheet corresponding to the breakup; because of the axial symmetry, the conservation of this covector gives the same condition at all points.}
\begin{equation}
P=-E\, \xi +L\, \upsilon +P_r\, \zeta \, ,
\end{equation}
where  $\{\xi,\upsilon,\zeta\}$ is the dual coframe to $\{ X,Y,Z\}$, with $X=\frac{\partial}{\partial t}$, $Y=\frac{\partial}{\partial \varphi}$ and $Z=\frac{\partial}{\partial r}$.

The naive generalization of the Penrose process for elastic string loops would be to consider a scenario where some string loop, with energy $E_0$, comes from infinity, breaks up inside the ergoregion into two new string loops, with energies $E_1$ and $E_2$, one of which then escapes, carrying energy $E_2$ to infinity. A successful Penrose process requires $\eta=\frac{E_2}{E_0}>1$; in particular we want this quantity to be as large as possible in order to have an efficient extraction of energy. However, this gedanken experiment faces a considerable problem in the case of elastic string loops: the fact that string loops with a finite relaxed radius cannot in general reach infinity. Hence, there is only one physically reasonable way to extend the typical Penrose process to an \emph{elastic} Penrose process: a ring of dust comes from infinity, interacts in the ergoregion to produce an elastic string loop and another ring of dust, which will take energy away to infinity. This choice solves the problem, and the only difference from the usual Penrose process is that the negative energy particle that remains in the ergoregion is replaced by an elastic string loop. It can be physically interpreted as a ring of noninteracting particles decaying into strongly bonded particles (the elastic string loop) plus more noninteracting particles (the dust ring that escapes to infinity).

Our main goal is to understand if it is possible to achieve higher efficiencies $\eta$ through this new process. The usual Penrose process maximizes the efficiency when a massive particle from infinity breaks into two photons in the ergoregion.\footnote{There are several works on the maximal efficiency of the Penrose process with test particles \cite{W74b, BPT72}, which have been extended to the collisional Penrose process \cite{PS77, Schnittman14, BBC15, LP16}.}
 The \emph{elastic Penrose process} will only be able to achieve higher efficiencies if the (effective) $4$-momentum of the elastic string loop is spacelike at the breakup. In fact, since the incoming and outgoing string loops are dust rings, and so behave like particles, and since the only role played by the elastic string loop is in $4$-momentum conservation at the breakup, whenever the $4$-momentum of the elastic string loop is causal the process will emulate a particle process.
\begin{Thm}\label{Thm_Timelike}
Outside the event horizon, the effective $4$-momentum $P$ for elastic string loops satisfying the dominant energy condition is causal and future-pointing.
\end{Thm}
\begin{proof}
From \eqref{conserved_quantity}, we have 
\begin{equation}
-E = \int_{\{\tau=\text{constant}\}} T^{AB}X_A \nu_B \sqrt{h_{11}} \, d\lambda  \, ,
\end{equation} 
\begin{equation}
L = \int_{\{\tau=\text{constant}\}} T^{AB}Y_A \nu_B \sqrt{h_{11}} \, d\lambda \, ,
\end{equation} 
\begin{equation}
P_r = \int_{\{\tau=\text{constant}\}} T^{AB}Z_A \nu_B \sqrt{h_{11}} \, d\lambda \, ,
\end{equation} 
Moreover, from the definition of $P$ 
\begin{equation}
\langle P,X \rangle = -E\, ,
\end{equation}
\begin{equation}
\langle P, Y \rangle = L\, , 
\end{equation}
\begin{equation}
\langle P, Z \rangle = P_r\, .
\end{equation}
Recall that at each point outside the event horizon there exists a quantity $\Omega_{ZAMO}$, called the {\em angular velocity of the zero angular momentum observers}, such that the vector
\begin{equation}
X + \Omega_{ZAMO} Y
\end{equation}
is timelike and future-pointing.\footnote{In fact $\Omega_{ZAMO}$ is chosen so that $X + \Omega_{ZAMO} Y$ is orthogonal to the spacelike hypersurfaces of constant $t$ in Boyer-Lindquist coordinates.}
Let us now consider the family of vectors
\begin{equation}
X^{(\Omega, \epsilon)}= X+\Omega Y + \epsilon Z\,.
\end{equation}
It is possible to choose $\Omega$ and $\epsilon$ in such a way that $X^{(\Omega, \epsilon)}$ sweeps out the whole cone of future causal directions at every point outside the event horizon. Assuming the dominant energy condition, we have
\begin{equation}
T^{AB}X^{(\Omega, \epsilon)}_A \nu_B \leq 0 \, ,
\end{equation}
because $\nu$ is timelike past-pointing, and so
\begin{equation}
\langle P, X^{(\Omega, \epsilon)} \rangle = -E + \Omega L + \epsilon P_r = \int_{\{\tau=\text{constant}\}} T^{AB}X^{(\Omega, \epsilon)}_A \nu_B \sqrt{h_{11}} \, d\lambda \leq 0 \, .
\end{equation}
Since $X^{(\Omega, \epsilon)}$ sweeps out the whole cone of future causal directions, this inequality implies that $P$ is causal future-pointing.
\end{proof}
We can summarize our result as follows:
\begin{Thm}\label{Thm_elasticprocess}
The elastic Penrose process with string loops satisfying the dominant energy condition cannot achieve higher efficiencies than the known bound for the usual Penrose process, given by $\eta_0 = \frac{\sqrt{2}+1}{2}$.
\end{Thm}

\subsection{Energy extraction and the dominant energy condition}
Let us now consider an elastic string loop that violates the dominant energy condition but preserves the weak energy condition. It is easy to see that an elastic string loop that violates the weak energy condition can have negative energy even in Minkowski's spacetime, and so it does not make sense to discuss the Penrose process in this case, since, in principle, we do not even need an ergoregion to accumulate energy at infinity. This is also the reason why it does not make sense to consider test particles that violate the dominant energy condition: such particles would also violate the weak energy condition. 

Let us then consider a string loop given by \eqref{rigid} with $c^2>1$. This string violates the dominant energy condition but preserves the weak energy condition. Numerically, setting $M=a=1$, $\dot{R}=0$ and $R=1.01$ at the breakup, we find that we need $\dot{\Phi} \in (0.492513, 0.497512)$ for the worldsheet to be timelike, that is, in order to have
\begin{equation}
-\left(1-\frac{2M}{R}\right)+\frac{R^2}{\Delta}\dot{R}^2+\left(R^2+a^2+\frac{2Ma^2}{R}\right)\dot{\Phi}^2 -\frac{4Ma}{R}\dot{\Phi}<0\,.
\end{equation}
Moreover, we set the following ratios: $\frac{\rho_0}{m_0}=10^{-4}$,  $\frac{m_2}{m_0}=10^{-9}$, where $m_0$ and $m_2$ are the effective rest masses of the incoming and outgoing dust rings of particles, respectively, and $\rho_0$ is the parameter in \eqref{rigid}. We solve the conservation equations \eqref{conservation_break} for $E_2$, the energy of the outgoing dust, and $k$, the inverse of the elastic string loop's relaxed radius, and we force $k>0$. By varying $\dot{\Phi}$ between its allowed values, we are able to find several numerical solutions for the efficiency which are represented in Figure~\ref{elastic_penrose}. We can see that some solutions correspond to higher efficiencies than the usual Penrose process, $\eta_0 = \frac{\sqrt{2}+1}{2}$. Actually, one can check that for these solutions the momentum vector $P$ of the string is spacelike, as required by the argument in the previous section. It is also possible to confirm that the outgoing ring of dust escapes to infinity.

\begin{figure}[h] 	
\centering 	
\includegraphics[scale=0.7]{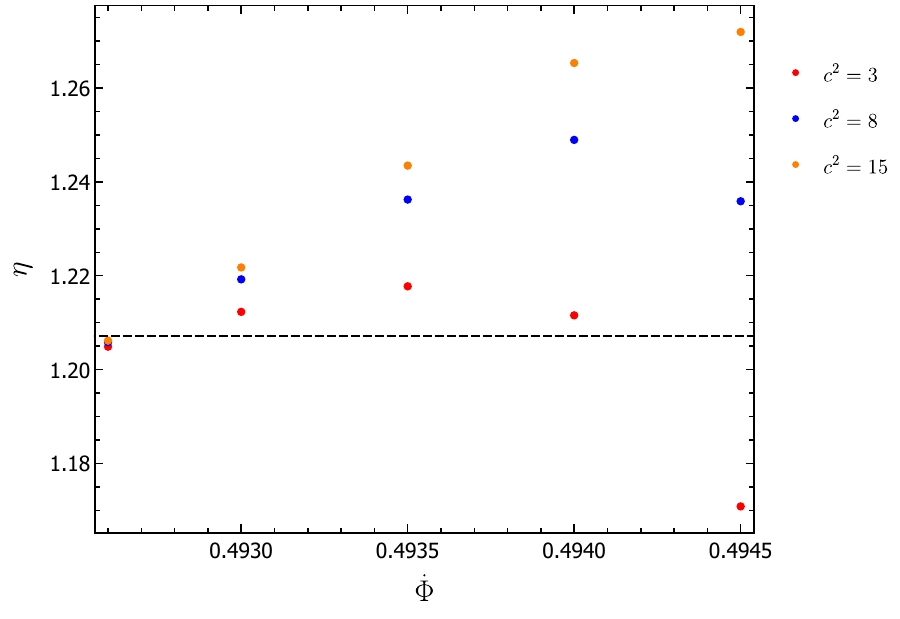} 
\caption{Plot of the efficiency of the \emph{elastic} Penrose process for strings with different values of $c^2>1$ that  violate the dominant energy condition. The dashed black line corresponds to the maximum efficiency $\eta_0 = \frac{\sqrt{2}+1}{2}$ of the standard Penrose process with test particles.} \label{elastic_penrose}
\end{figure}

Moreover, the larger we make $c^2$ the higher the efficiencies we can achieve are, and possibly this happens without bound when we let $c^2 \rightarrow +\infty$. It is hard to see if this indeed happens, because the equations become numerically unstable when $c^2$ becomes too large, but it seems a likely scenario.

\begin{Thm}\label{Thm_elasticprocess_1}
The elastic Penrose process with string loops violating the dominant energy condition (but satisfying the weak energy condition) allows efficiencies $\eta$ higher than the known bound for the usual Penrose process, given by $\eta_0 = \frac{\sqrt{2}+1}{2}$.
\end{Thm}

Although the worldsheets of the string loops we just considered are timelike, perturbations travel along these strings faster than the speed of light. Hence, elastic string loops that violate energy conditions are somewhat unphysical. Nevertheless, they allow us to see explicitly that the bound on the efficiency of the Penrose process is in some way related to the dominant energy condition. This last result hints at the interesting possibility that the dominant energy condition may underlie the well known upper bounds for the efficiencies of energy extraction processes. In particular, one can wonder whether it is possible to obtain unbounded superradiance \cite{BCP15} with fields that violate the dominant energy condition, since superradiance is usually seen as the field analog of the Penrose process.

%
%
\section{Conclusion}
To conclude, we briefly summarize our results and discuss future directions of research. 

\subsection{Formalism}

In the first part of this work we re-derived the equations of motion for elastic strings starting from a Lagrangian density, and recast these equations as conservation of energy-momentum along the worldsheet plus the {\em generalized sail equations} (the vanishing of the contraction between the energy-momentum tensor and the extrinsic curvatures of the worldsheet). We also obtained the conserved quantities in spacetimes with Killing vector fields and computed the string's longitudinal and transverse speeds of sound. Although these results are known, it is convenient to have them deduced from first principles and collected together.

\subsection{Equilibria and stability}

In the second part of this work we determined the equilibrium configurations for axially symmetric string loops rotating in the equatorial plane of the Kerr spacetime. We proved that in flat spacetime the radius of a rotating string loop is always larger than its radius when relaxed, provided that its longitudinal speed of sound does not exceed the speed of light. Still in flat spacetime, we showed that equilibria are linearly stable for rotation speeds below a certain threshold, higher than the string's longitudinal speed of sound, and linearly unstable for some rotation speeds above it, confirming and extending the results in \cite{CM93, Martin94}. Finally, we found that equilibria are always linearly unstable in Schwarzschild's spacetime.

It would be interesting to extend our stability analysis to the Kerr spacetime. Although we expect similar results to the Schwarzschild case, it is conceivable that the frame dragging associated to the black hole's rotation may have a stabilization effect.

\subsection{Weak cosmic censorship and the null energy condition}

In the third part of this work we tried to destroy an extremal black hole by letting it absorb a string loop with an arbitrary elastic law. We found that this gedanken experiment can only succeed if the string violates the null energy condition at the event horizon. This result is in agreement with \cite{NQV16, SW17}, where it was shown that test fields satisfying the null energy condition at the event horizon cannot destroy extremal black holes. Since localized matter can be seen as a limit of test fields, one can argue that this conclusion should extend to general test matter. Indeed, our work shows that the same statement holds for some extended objects, bridging the gap bewteen the recent developments for test fields and the original results for test particles \cite{W74}.

\subsection{Penrose process and the dominant energy condition}

Finally, we tried to increase the efficiency of the standard Penrose process by using elastic string loops. We found that if the loops satisfy the dominant energy condition then we cannot get higher efficiencies than that of the standard Penrose process \cite{W74b, BPT72}. Nevertheless, if we consider elastic loops that violate the dominant energy condition but preserve the weak energy condition then we seem to be able to get arbitrarily high efficiencies for the Penrose process. 

Although classical matter that violates the dominant energy condition is not very interesting from the astrophysical point of view, this result may give us a better understanding on the principles underlying energy extraction from black holes. Is the dominant energy condition imposing constraints on the amount of energy that we can extract? Still along these lines, we may wonder what happens if we consider the superradiance \cite{BCP15} produced by a field that violates the dominant energy condition but preserves the weak energy condition. Will we be able to get unbounded superradiance in analogy with the \emph{elastic} Penrose process?

%
\section*{Acknowledgments}
We thank Prof.~Ant\'{o}nio Brotas for kindly providing a copy of his 1969 PhD thesis on relativistic thermodynamics and continuum mechanics (under Louis de Broglie), as well as some of his early papers. We also thank Jo\~{a}o Penedones for hypothesizing a connection between the radius of a rotating string loop and its speed of sound, and Vitor Cardoso for his suggestion of using extended objects to probe the cosmic censorship conjecture. JN was partially funded by FCT/Portugal through projects UID/MAT/04459/2013 and PTDC/MATANA/1275/2014. LQ gratefully acknowledges a scholarship from the Calouste Gulbenkian Foundation program {\em Novos Talentos em Matem\'atica}, as well as an undergraduate research grant from the FCT/Portugal project PTDC/MATANA/1275/2014. RV was supported by graduate research fellowships from the FCT/Portugal projects EXCL/MAT-GEO/0222/2012 and UID/FIS/00099/2013.
\appendix
%
%
%
\section{Non-Relativistic Limit} \label{appendixA}
It is instructive to compute the non-relativistic limit of the Lagrangian density \eqref{Lagrangian} in the case of flat Minkowski spacetime. We start by defining the embedding
\begin{equation}
\begin{cases}
t = \tau \\
x^i = x^i(\tau,\lambda)
\end{cases}.
\end{equation}
Substituting in the Minkowski metric
\begin{equation}
ds^2 = - dt^2 + \delta_{ij} dx^i dx^j
\end{equation}
we obtain
\begin{equation}
(h_{AB}) = 
\left(
\begin{matrix}
-1 + \dot{\bf x}^2 & \dot{\bf x} \cdot {\bf x}' \\
\dot{\bf x} \cdot {\bf x}' & {{\bf x}'}^2
\end{matrix}
\right),
\end{equation}
where we have set
\begin{equation}
{\bf x}(t,\lambda) = (x^1(t,\lambda), x^2(t,\lambda), x^3(t,\lambda)), \qquad \dot{\bf x} = \frac{\partial{\bf x}}{\partial t}, \qquad {\bf x}' = \frac{\partial{\bf x}}{\partial \lambda}.
\end{equation}
Therefore we have
\begin{equation}
n^2 = \frac{h_{00}}{h} = \frac{1 - \dot{\bf x}^2}{{{\bf x}'}^2(1 - \dot{\bf x}^2)+(\dot{\bf x} \cdot {\bf x}')^2}= \frac1{{{\bf x}'}^2}\left( 1 - \frac{(\dot{\bf x} \cdot {\bf t})^2}{1 - \dot{\bf x}^2 + (\dot{\bf x} \cdot {\bf t})^2}\right),
\end{equation}
where ${\bf t} = {\bf x}'/\|{\bf x}'\|$ is the unit tangent vector to the string. The Lagrangian density \eqref{Lagrangian} can then be written as
\begin{equation}\label{Lagrangian2}
\cL = F(n^2) \sqrt{-h} = F(n^2) \|{\bf x}'\| \sqrt{1 - \dot{\bf x}^2 + (\dot{\bf x} \cdot {\bf t})^2}.
\end{equation}
To obtain the nonrelativistic limit, we assume that
\begin{equation}
\|\dot{\bf x}\| \sim \varepsilon, \qquad \|{\bf x}'\| \sim 1 + \varepsilon
\end{equation}
and retain only terms up to order $\varepsilon^2$ in the expression of the Lagrangian density. Note that in particular
\begin{equation}
n^2 \sim \frac1{{{\bf x}'}^2} - (\dot{\bf x} \cdot {\bf t})^2 \sim 1 + \varepsilon,
\end{equation}
and so
\begin{equation}
F(n^2) \sim F(1) + F'(1) (n^2 - 1) + \frac12 F''(1) (n^2 - 1)^2,
\end{equation}
or, using \eqref{p(1)=0} and \eqref{c^2},
\begin{equation}
F(n^2) \sim \rho_0 \left( 1 + \frac12 (n^2 - 1) + \frac{{c_0}^2 - 1}8 (n^2 - 1)^2 \right),
\end{equation}
where $\rho_0$ and $c_0$ are the density and the speed of sound of the relaxed string. Dropping the multiplicative constant $\rho_0$, we can write the Lagrangian density \eqref{Lagrangian2} up to order $\varepsilon^2$ as
\begin{align}
\cL & \sim \frac12 \left( 1 + \frac1{{{\bf x}'}^2} - (\dot{\bf x} \cdot {\bf t})^2 + \frac{{c_0}^2 - 1}4 \left(\frac1{{{\bf x}'}^2} - 1\right)^2 \right) \|{\bf x}'\| \left( 1 - \frac12 \dot{\bf x}^2 + \frac12 (\dot{\bf x} \cdot {\bf t})^2\right) \nonumber \\
& \sim  \frac12 \left( \|{\bf x}'\| + \frac1{\|{\bf x}'\|} - \dot{\bf x}^2 + (\dot{\bf x} \cdot {\bf t})^2 - (\dot{\bf x} \cdot {\bf t})^2 + \frac{{c_0}^2 - 1}4 \left(\frac1{{{\bf x}'}^2} - 1\right)^2 \right).
\end{align}
Noting that up to order $\varepsilon^2$
\begin{equation}
\|{\bf x}'\| + \frac1{\|{\bf x}'\|} = \frac{{{\bf x}'}^2 + 1}{\|{\bf x}'\|} = 2 + \frac{{{\bf x}'}^2 - 2\|{\bf x}'\| + 1}{\|{\bf x}'\|} \sim 2 + \left(\|{\bf x}'\| - 1 \right)^2
\end{equation}
and
\begin{equation}
\left(\frac1{{{\bf x}'}^2} - 1\right)^2 = \left(\frac{{{\bf x}'}^2-1}{{{\bf x}'}^2}\right)^2 = \frac{(\|{\bf x}'\| - 1)^2(\|{\bf x}'\| + 1)^2}{\|{\bf x}'\|^4} \sim 4 \left(\|{\bf x}'\| - 1 \right)^2,
\end{equation}
we obtain
\begin{equation}
\cL \sim  \frac12 \left( 2 - \dot{\bf x}^2 + {c_0}^2 \left(\|{\bf x}'\| - 1 \right)^2 \right),
\end{equation}
or, dropping the additive constant and multiplying by $-1$,
\begin{equation}
\cL \sim  \frac12 \dot{\bf x}^2 - \frac{{c_0}^2}2 \left(\|{\bf x}'\| - 1 \right)^2,
\end{equation}
which is the Newtonian Lagrangian that one would expect for an elastic string with longitudinal speed of sound $c_0$.

One can study axially symmetric rotating string loops using the Lagrangian above, possibly including a gravitational potential of the form $-\frac{M}{\|{\bf x}\|}$. It turns out that equilibrium configurations are always stable in the absence of gravity for nonzero rotation speeds,\footnote{Note that the assumption $\|{\bf x}'\| \sim 1 + \varepsilon$ leads to a transverse speed of sound $s \sim \varepsilon \ll c_0$.} whereas they are always unstable if gravity is included.
%
%
%
\section{Equilibrium configurations in stationary spacetimes} \label{appendixB}
In this appendix we obtain the equilibrium conditions of elastic strings in stationary spacetimes to explain the strange behaviour of stationary strings in the Kerr spacetime. See~\cite{FSZ89, CF89, Carter90, CFH91, Larsen91, FS04, KF08, AA09} for related (but different) formulations.

The metric of a stationary spacetime $(M,g)$ can be written as
\begin{equation}
g = - e^{2\phi} (dt + A_i dx^i)^2 + \gamma_{ij} dx^i dx^j,
\end{equation}
where $\frac{\partial}{\partial t}$ is the timelike Killing vector field, so that $\phi$, $A_i$ and $\gamma_{ij}$ do not depend on $t$. For a stationary string, we can assume without loss of generality that the worldsheet $\Sigma$ is given by $x^2=\ldots=x^n=0$, and that on $\Sigma$
\begin{equation}
g_{|_\Sigma} = - e^{2\phi} dt^2 + \delta_{ij} dx^i dx^j.
\end{equation}
Therefore, the extrinsic curvatures are
\begin{equation}
K^i = - \partial_i \phi \, e^{2\phi} dt^2 + e^{2\phi} \partial_i A_1 dt dx^1 + \widetilde{K}^i_{11} \left(dx^1\right)^2 \qquad (i=2,\ldots,n),
\end{equation}
where $\widetilde{K}^i$ are the extrinsic curvatures of the curve $x^2=\ldots=x^n=0$ in the {\em space manifold}, that is, the $n$-dimensional Riemannian manifold with coordinates $(x^1, \ldots, x^n)$ and metric
\begin{equation}
\gamma = \gamma_{ij} dx^i dx^j.
\end{equation}
Note that each point of the space manifold corresponds to a stationary observer. The Riemannian metric $\gamma$ gives the distances between neighboring stationary observers, as measured in radar experiments, for instance. 

The metric of the worldsheet is given by
\begin{equation}
h = -e^{2\phi} dt^2 + \left(dx^1\right)^2,
\end{equation}
and so its energy-momentum tensor can be written as
\begin{equation}
T = \rho e^{2\phi} dt^2 + p \left(dx^1\right)^2,
\end{equation}
leading to the generalized sail equations
\begin{equation} \label{equilibrium_stationary}
-\rho\, \partial_i \phi + p \widetilde{K}^i_{11} = 0 \qquad (i=2,\ldots,n).
\end{equation}
Since the string is stationary, the $4$-velocity of the its particles is simply
\begin{equation}
U = e^{-\phi} \frac{\partial}{\partial t},
\end{equation}
and it is easy to compute that on $\Sigma$
\begin{equation}
\nabla_U U = \partial_i \phi \frac{\partial}{\partial x^i}.
\end{equation}
We can think of
\begin{equation}
G = - \nabla_U U
\end{equation}
as the gravitational field measured by the stationary observers, and so \eqref{equilibrium_stationary} can be interpreted as stating that the orthogonal component of the gravitational force is being balanced by the tension of the string times its geodesic curvature:\footnote{Note that with our conventions the geodesic curvature vector points away from the center of curvature.}
\begin{equation}
\rho\, G^i + p \widetilde{K}^i_{11} = 0 \qquad (i=2,\ldots,n).
\end{equation}

For a stationary string, $\rho$ and $p$ do not depend on $t$. Therefore, in the equations of conservation of energy-momentum along the worldsheet,
\begin{equation}
\overline{\nabla}_A \left[(\rho+p) U^A U^B + p h^{AB} \right] = 0,
\end{equation}
only the component orthogonal to the Killing vector field is nontrivial:
\begin{equation} \label{equilibrium_stationary_1}
(\rho + p) G^1 - \frac{dp}{d x^1} = 0.
\end{equation}
We can think of this equation as stating that the tangential component of the gravitational force is being balanced by the variation in the tension of the string.

\subsection{Stationary loops in the Kerr spacetime}
It is easy to obtain from the Kerr metric in Boyer-Lindquist coordinates that in the equatorial plane (outside the ergosphere) the space metric is given by
\begin{equation}
\gamma = \frac{r^2}{\Delta} dr^2 + \frac{\Delta}{V} d\varphi^2,
\end{equation}
where $V=1-\frac{2M}{r}$. Since
\begin{equation}
\frac{d}{dr} \frac{\Delta}{V} = \frac{2}{r^2 V} \left( r^3 V^2 - M a^2 \right),
\end{equation}
we conclude that the circles of constant $r$ curve inwards for large $r$, are {\em geodesic} for $r^3 V^2 - M a^2=0$, and curve {\em outwards} for smaller values of $r$.\footnote{This is reminiscent of the behavior of the optical metric near the photon sphere of the Schwarzschild solution, used in \cite{AL86} to explain the reversal of the centrifugal force discovered in \cite{AL88}.} This explains why the compression of a stationary axially symmetric string loop in the equatorial plane of the Kerr spacetime blows up for $r^3 V^2 - M a^2=0$, and becomes a tension for smaller values of $r$.

\subsection{Lowering a mass using an elastic string}
An obvious equilibrium configuration is to place the string along a geodesic of the space metric $\gamma$ tangent to the gravitational field (for example radial geodesics in the case of the Kerr equatorial plane). It is interesting to consider the case when the spacetime is asymptotically flat and the string is being used to hang a mass $m$ from some point at spatial infinity. From~\eqref{equilibrium_stationary_1} we have
\begin{equation} \label{equilibrium_stationary_2}
\frac{d\phi}{d x^1} = -\frac1{\rho + p}\frac{dp}{d x^1},
\end{equation}
and from~\eqref{rhop} we obtain
\begin{equation}
\frac{dp}{d x^1} = \left[2n^2F''(n^2)+F'(n^2)\right]\frac{dn^2}{d x^1}.
\end{equation}
Substituting in \eqref{equilibrium_stationary_2} and integrating yields
\begin{equation} \label{potential}
e^{\phi(x^1)} = \frac{F'(n^2(\infty)) n(\infty)}{F'(n^2(x^1)) n(x^1)},
\end{equation}
where we have set $\phi(\infty)=0$.

The energy obtained at infinity by slowly lowering the mass $m$ from infinity to some finite position $x^1$ is given by the work done at infinity (on a pulley, say) by the force $-p(\infty)$, that is,
\begin{equation}
W = \int_{x^1}^{\infty} -p(\infty) \frac{n(x)}{n(\infty)} dx,
\end{equation}
where the correction factor $\frac{n(x)}{n(\infty)}$ is the ratio between the string's stretching at infinity and its stretching at $x$; note carefully that both $p(\infty)$ and $n(\infty)$ in this integral depend on $x$. Using~\eqref{rhop} and \eqref{potential} we obtain,
\begin{align}
W & = \int_{x^1}^{\infty} -\left[2n^2(\infty) F'(n^2(\infty)) - F(n^2(\infty)) \right] \frac{n(x)}{n(\infty)} dx \nonumber \\
& = \int_{x^1}^{\infty} \left[-2n^2(x) F'(n^2(x)) e^{\phi(x)} + F(n^2(\infty)) \frac{n(x)}{n(\infty)} \right] dx \label{integral} \\
& = \int_{x^1}^{\infty} -p(x) e^{\phi(x)} dx + \int_{x^1}^{\infty} \left[F(n^2(\infty)) \frac{n(x)}{n(\infty)} - F(n^2(x)) e^{\phi(x)} \right] dx. \nonumber 
\end{align}
Since the mass is lowered quasi-statically, and so is in equilibrium for each value of $x$, we must have
\begin{equation}
-p(x) - m \frac{d\phi}{dx} = 0.
\end{equation}
Therefore the first term in \eqref{integral} is simply
\begin{equation}
W_m = \int_{x^1}^{\infty} m \frac{d\phi}{dx} e^{\phi(x)} dx = m \left(1 - e^{\phi(x^1)}\right),
\end{equation}
and can be interpreted as the potential energy lost by the mass. The second term must then be the potential energy lost by the string (which has mass itself). It is interesting to note that for strings with constant longitudinal speed of sound this energy does not depend on $m$. Indeed, if we write \eqref{rigid} as
\begin{equation}
F(n^2) = \frac{\rho_0}{c^2+1} \left((n^2)^\frac{c^2+1}2 + c^2\right),
\end{equation}
then \eqref{potential} yields
\begin{equation}
n(\infty) = e^{\frac{\phi(x)}{c^2}} n(x),
\end{equation}
and so we can write the second term in \eqref{integral} as\footnote{Note that in general $W_s$ may be infinite, but this can be easily remedied by replacing $x^1=\infty$ with $x^1=L$ for sufficiently large $L$.}
\begin{equation}
W_s = \int_{x^1}^{\infty} \left[F(n^2(\infty)) \frac{n(x)}{n(\infty)} - F(n^2(x)) e^{\phi(x)} \right] dx = \frac{\rho_0 c^2}{c^2 + 1} \int_{x^1}^{\infty} \left(1-e^{\phi(x)}\right) dx. 
\end{equation}
%
%
%
\section{Coefficients for the stability analysis} \label{appendixC}

The coefficients of equations \eqref{motion_linear1}, \eqref{motion_linear2} and \eqref{motion_linear3} are:

\begin{align*}
& A=- \frac{(R-M) (s^2+c^2 ) \sqrt{(R-2M)s^2+M}}{k (R-2M)^{3/2} \sqrt{R^3 \left[ (R-2M)+M s^2 \right]}}; \\
& B=- \frac{c^2 R (1-s^2) \sqrt{(R-2M)s^2+M}}{k^2 \sqrt{R-2 M} \sqrt{R^3 \left[(R-2M) +M s^2\right]}}; \\
& C=- \frac{\left[(R-2M)s^2+M \right] \left[(R-M)c^2 +(R-3M)s^2-2(R-2M) \right]}{(R-3 M) (R-2 M)^2 (1-s^2 )}; \\
& D=- \frac{2 c^2 R \left[(R-2M)s^2+M \right]}{k (R-2 M) (R-3 M)}; \\
& E=- \frac{R \sqrt{R^3 \left[(R-2M)+M s^2 \right]} \sqrt{(R-2M)s^2+M}\, \left[M (c^2-s^2)-(R-2 M) (1-c^2 s^2)\right]}{(R-2 M)^{3/2} (R-3 M)^2 (1-s^2 )}; \\
& F= \frac{ c^2\left(-3 M^3+7 M^2 R-5 M R^2+ R^3\right)-3 M (R-2M)^2-M s^4 \left(12 M^2-8 M R+R^2\right)-s^2  (R-3 M)^3}{R^3 (R-3 M) (R-2 M) (1-s^2 )}; \\
& G= \frac{(R-M) (s^2+c^2 )}{k R^2}; \\
& H=- \frac{s^2 (R-3 M) (1-s^2)}{k^2 R^2 \left[(R-2M)+M s^2\right]}; \\
& I= \frac{\sqrt{R^3 \left[(R-2M)+M s^2\right]}\, \sqrt{(R-2M)s^2+M} \left[c^2(R-M)+s^2 (R-3M)-2(R-2M) \right]}{R^2 \sqrt{R-2 M}\, (R-3 M) (1-s^2)}; \\
& J= - \frac{2 s^2 R \sqrt{(R-2M)s^2+M}}{k \sqrt{R-2 M}\, \sqrt{R^3 \left[(R-2M)+M s^2\right]}}; \\
& L= \frac{R (1+s^2)}{R-3 M}; \\
& N= \frac{M (1+s^2)}{R^2 (R-3 M)} .
\end{align*}
The quantities $s$ and $c$ are taken at the particular equilibrium under consideration.

\end{document}